\documentclass[12pt]{llncs}
\usepackage{lineno}
\usepackage{amsmath,amssymb,amsmath}
\pdfoutput=1
\usepackage{ifpdf}
\usepackage{graphicx}
\usepackage{color}
\usepackage{fullpage}
\usepackage{textcomp}
\usepackage{setspace}
\usepackage{algorithmic}
\usepackage[ruled]{algorithm}
\usepackage{ulem}
\newtheorem{observation}{Observation}

\DeclareMathAlphabet{\mathpzc}{OT1}{pzc}{m}{it}

\title{Localized Geometric Query Problems}
\author{John Augustine\inst{1} \and Sandip Das \inst{2} \and 
Anil Maheshwari \inst{3}\thanks{Research supported by NSERC.} \and Subhas C. 
Nandy 
\inst{2} \and 
Sasanka Roy\inst{4} \and  \\ Swami Sarvattomananda \inst{5}}

\institute{Department of Computer Science and Engineering, Indian Institute of 
Technology 
Madras, Chennai, India
\and  Advanced Computing and Microelectronics Unit, 
Indian Statistical Institute, Kolkata, India
\and  School of Computer Science, Carleton University, 
Ottawa, Canada
\and Chennai Mathematical Institute, Chennai,India
\and  School of Mathematical Sciences, Ramakrishna 
Mission Vivekananda University, Belur, India}

\begin{document}
\maketitle

\begin{abstract}
A new class of geometric query problems are studied in this paper. 
We are required to preprocess a set of geometric objects $P$ in 
the plane, so that for any arbitrary query point $q$, the largest 
circle that contains $q$ but does not contain any member of $P$, 
can be reported efficiently. The geometric sets that we consider 
are point sets and boundaries of simple polygons. 
\end{abstract}


\vspace{-0.2in}
\section{Introduction}
Largest empty space recognition is a classical problem in computational 
geometry, and has applications in several disciplines like database 
management, operations research, wireless sensor network, VLSI, to name 
a few. Here the problem is 
to identify an empty space of a desired shape and of maximum size in a  
region containing a set of obstacles. Given a set $P$ of $n$ 
points in $\mathbb{R}^2$, an {\it empty circle}, is a circle that does 
not contain any member of $P$. An empty circle is said to be a {\it 
maximal empty circle} ({\tt MEC}) if it is not fully contained in any 
other empty circle. Among the {\tt MEC}s, the one having the maximum radius 
is the {\it largest empty circle}. The largest empty circle among a point 
set $P$ can easily be located by using the Voronoi diagram of $P$ in 
$O(n\log n)$ time \cite{T}.

Although a lot of study has been made on the empty space recognition 
problem, surprisingly, the query version of the problem  has not received 
much attention. The problem of finding the 
largest empty circle centered on a given query line segment has been 
considered in \cite{APS10}. The preprocessing time, space and query time 
complexities of the algorithm in \cite{APS10} are $O(n^3\log n)$, $O(n^3)$ and 
$O(\log n)$, respectively. In practical applications, one may need to 
locate the largest empty circle in a desired location. For 
example, in the VLSI physical design, one may need to place a large circuit 
component in the vicinity of some already placed components. Such problems 
arise in mining large data sets as well, where one of the objectives is to  
search for empty spaces in  data sets\ \cite{LKH1997}. In\ \cite{EGLM2003}, 
Edmonds {\it et al.} formalized the problem of finding  large empty spaces in 
geometric data sets. In particular, they studied the problem of finding large empty 
rectangles in data sets.

An important problem in this context is the {\it circular 
separability problem}. Two planar sets $P_1$ and $P_2$ are circularly 
separable if there is a circle that encloses $P_1$ but excludes $P_2$. 
O'Rourke {\it et al.} \cite{OKM86} showed that the decision version of the circularly 
separability of two sets can be solved in $O(n)$ time using linear 
programming. Furthermore, they show that a smallest separating circle can be found 
in $O(n)$ time while the computation of the largest separating circle needs 
$O(n \log n)$ time. Detailed study on circular separability problem can be 
found in \cite{BCDUY00,BCOY01,Fisk86,OKM86}. Boissonnat {\it et al.} \cite{BCOY01}
proposed a linear-time algorithm for solving the decision version of
circular separability problem where the sets $P_1$ and $P_2$ are simple
polygons, and the algorithm outputs the smallest separating circle. They also 
consider the query version of this problem where the objective is to
preprocess a convex polygon $P$ such that given a query point $q$ and a query
line $\ell$,  report the largest circle inside $P$ that contains $q$ and
does not intersect $\ell$. The preprocessing time and space
complexities of their proposed algorithm are both $O(n \log n)$, and the query
can be answered in $O(\log n)$ time. They also showed that a convex polygon
$P$ can be preprocessed in $O(n)$ time and space such that for a query set $S$
of $k$ points, the largest circle inside $P$ that encloses $S$ can be computed
in $O(k \log n)$ time.

In addition to empty circles, empty rectangles have also been studied. We introduced the query version of the maximal empty rectangle in \cite{ADMNRS10}. The problem entails preprocessing a set of $n$ points such that, given a query point $q$, the largest empty rectangle containing $q$ can be reported efficiently. We gave a solution with  query time  $O(\log n)$ with preprocessing time and space being $O(n^2 \log n)$ and $O(n^2)$, respectively.
Recently, Kaplan {\it et al.} \cite{KMNS12} improved the preprocessing time 
and space complexities  to $O(n\alpha(n)\log^4 n)$ and 
$O(n\alpha(n)\log^3 n)$, respectively, while the query time has increased to 
$O(\log ^4 n)$. Here $\alpha(n)$ is the inverse Ackermann function.

\subsection{Our Results}
In this paper, we  study the query versions of the maximum empty circle  problem ({\tt QMEC}).  
The following variations are considered. 
\begin{itemize}
\item Given a simple polygon $P$, preprocess it 
such that given a query point $q$, the largest circle inside $P$ 
that contains the query point $q$ can be identified efficiently.
\item Given a set of points $P$, preprocess it such 
that given a query point $q$, the largest circle that does not contain 
any member of $P$, but contains the query point $q$ can be identified 
efficiently.
\end{itemize}
Our results are summarized in Table 
\ref{TAB1}.
\begin{table} \caption{Complexity results for different variations 
of largest empty space query problem}
\begin{center}
\begin{tabular}{|l|c|c|c|} \hline 
Geometric set  & Preprocessing & Space & Query
time \\ 
&  time  &  & \\ \hline 
Simple Polygon &  $O(n \log^2 n)$ & $O(n \log n)$ & 
$O(\log n)$ \\ \hline 
Point Set &  $O(n^{3/2} \log^2 n)$ & $O(n^{3/2} \log n)$ & $O(\log 
n\log\log n)$ 
\\ \hline
Point Set &  $O(n^{5/3} \log n)$ & $O(n^{5/3})$ & $O(\log 
n)$ 
\\ \hline
\end{tabular}
\end{center}
\label{TAB1}
\vspace{-0.15in}
\end{table}

We believe that our work will motivate the study of new types of geometric 
query problems and may lead to a very active research area. The main 
theme of our work is to achieve  subquadratic preprocessing time and space, while ensuring 
polylogarithmic query times. 
The results in this paper, improve upon the results in our previous 
work  \cite{ADMNRS10}. Very recently, Kaplan and Sharir\ \cite{KS12} provided a solution to the {\tt QMEC} problem for point sets that only requires $O(n \log^2 n)$ time and $O(n \log n)$ space for preprocessing. Their query times, however, are $O(\log^2 n)$.

\subsection{Organization of the paper}
In Section \ref{sec:Convex}, as a preliminary requisite, we describe a way to answer {\tt QMEC} query for the case of convex polygons. The same bounds have been achieved by  Boissonnat {\it et al.} 
\cite{BCOY01}, but
our solution is slightly different and serves as the basis for solving the {\tt QMEC} problem on simple 
polygons. In Section \ref{sec:simple-polygon-case}, 
we present the {\tt QMEC} problem for simple polygons with $n$ vertices. The preprocessing 
time and space complexities are $O(n\log^2n)$ and $O(n\log n)$ respectively, 
and the query answering time is $O(\log n)$. In Section \ref{sec:QMEC}, we 
consider the same problem on  a set $P$ of $n$ points in $\mathbb{R}^2$. We present 
two algorithms (cf. Table\ \ref{TAB1}). Our first algorithm uses the 
concept of planar separators \cite{LT79} on the underlying planar graph corresponding to 
the  Voronoi  diagram of  $P$. It solves the {\tt QMEC} 
problem  on $P$ with $O(n^{3/2} \log^2 n)$ preprocessing time 
and $O(n^{3/2} \log n)$ space. Here, the queries can be answered in $O(\log n 
\log\log n)$ time. Our second algorithm (cf. Section\ \ref{sec:fastquery}) uses the $r$-partitioning \cite{fed} of planar graphs. With a suitable 
choice of
$r$, the query time is only $O(\log n)$, an improvement over our first algorithm. However, the preprocessing time 
and 
space increase
to $O(n^{5/3}\log n)$  and $O(n^{5/3})$, respectively. 

\section{Preliminaries: {\tt QMEC} problem for a convex polygon}
 \label{sec:Convex}

Let $P$ be a convex polygon and $\{p_1, p_2, \ldots, p_n\}$ be its
vertices in counter-clockwise order. Our objective in this section is to preprocess
$P$ such that given an arbitrary query point $q \in P$, the largest
circle $C_q$ containing $q$ inside the polygon $P$ can be reported 
efficiently. Note that, the locus of the centers of all the maximally empty circles ({\tt MEC}s) 
inside $P$ is defined to be the {\em medial axis} $M$ of $P$. Let $c$ be the center of 
the largest {\tt MEC} inside $P$ (see Figure \ref{fig:figure1}(a)).\footnote{There can be 
infinitely many {\tt MEC}s of largest radius, in which case we pick 
$c$ to be the center of one such {\tt MEC}.} The medial 
axis of a convex polygon consists of straight line
segments and can be viewed as a tree rooted at $c$ \cite{CSW99}. To
avoid  confusion with the vertices of the polygon, we call the
vertices of $M$ as nodes. Note that, the leaf-nodes of $M$ are the
vertices of $P$. Let us denote an {\tt MEC}  of $P$ centered at a point $x
\in M$ as ${\tt MEC}_x$ and let $A_x$ be the area of ${\tt MEC}_x$. 

In \cite{BCOY01}, a planar map of circular arcs is constructed by drawing the {\tt 
MEC} at 
each node of $M$ in $O(n)$ time and space.
The problem of finding $C_q$ reduces to the point location problem in
the associated planar map.  These point location queries can be answered in  $O(\log 
n)$ 
time. We propose an 
alternative solution (with the same complexity results as in\ \cite{BCOY01})  because our new technique plays a basic role in solving 
the problem when $P$ is a simple 
polygon (cf.  Section\ \ref{subsec:MiM}). We use the fact that the medial 
axis 
$M$
is a tree, and then use the level-ancestor queries \cite{BF04} on $M$. \begin{figure}
[thb]
\centering  
\includegraphics[width=\textwidth]{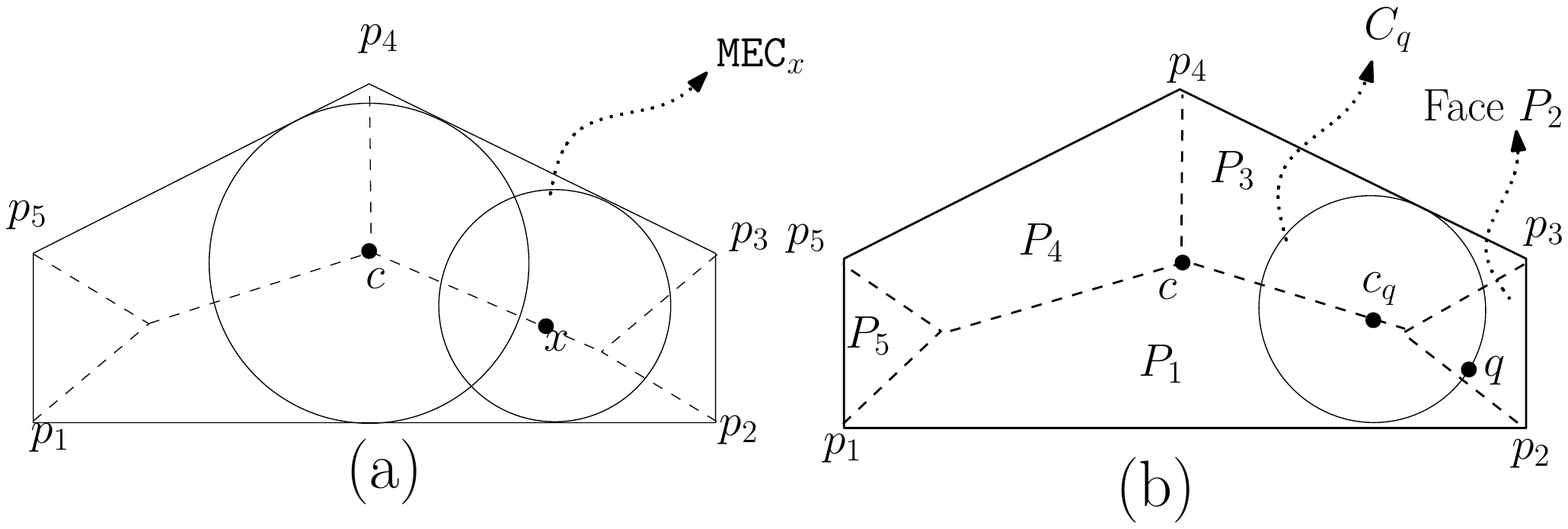}
  \caption{(a) Illustration of Lemma \ref{obs1}, and  (b)
partition of $P$.}
  \label{fig:figure1}
\end{figure}

\begin{lemma}\cite{BCOY01} \label{obs1}
As the point $x$ moves from the center $c$ of the largest MEC  along the medial axis 
towards any
vertex $p_i \in P$ (leaf node of $M$), $A_x$ decreases
monotonically.
\end{lemma}

\begin{lemma}\label{llx}
The polygon $P$ can be preprocessed in $O(n)$ time such that given any
arbitrary query point $q$, a point $x$ on $M$ such that ${\tt MEC}_x$ contains $q$ 
can be 
reported in $O(\log n)$ time. 
\end{lemma}
\begin{proof}
The medial axis $M$ subdivides  $P$ into $n$ convex faces 
such that each face $P_i$ consists of an edge $p_ip_{i+1}$ from $P$
and a convex chain of edges from $M$ connecting $p_i$ to $p_{i+1}$ (see Figure 
\ref{fig:figure1}(b)). In the preprocessing phase, we perform the following steps.
\begin{enumerate}
\item Compute the medial axis $M$ of $P$, which is a tree rooted at $c$.
\item Compute the subdivision  in $O(n)$ time. For this we will need $M$, which 
can be computed in linear time \cite{CSW99}. 
\item Store the chain of edges associated with each face in an array so that it is 
amenable to binary searching.
\item Finally, the subdivision can be 
preprocessed in $O(n)$ time so that the face containing a query point 
$q$ can be located in $O(\log n)$ time \cite{Krik83}.  
\end{enumerate}
 
In the query phase, we perform the following steps.
\begin{enumerate}
\item  We find the face $P_i$ that contains $q$ in $O(\log n)$ time. 

\item Recall that exactly one edge
$p_i p_{i+1}$ of $P$ will be an edge in $P_i$.  Consider the line $\ell$ through $q$ 
that is perpendicular to the edge $p_i p_{i+1}$. It will intersect an edge in $M$ that is 
also an edge bounding $P_i$;  we report that intersection point 
as $x$. Note that $x$ can be computed in $O(\log n)$ time via binary searching over the chain of medial axis edges bounding $P_i$.
\end{enumerate}
We need to  prove that ${\tt MEC}_x$ indeed encloses $q$.  Firstly, note that $\ell$ 
will  intersect the edge $p_i p_{i+1}$ internally at a point $t$ because (i) $P_i$ is 
convex and (ii) the two internal angles in $P_i$ at $p_i$ and $p_{i+1}$ are both acute.  
Secondly, note that any {\tt MEC} that goes through $t$ must be tangential to $p_i 
p_{i+1}$, thereby making it unique and centred on $\ell$; more precisely, the {\tt MEC} 
that goes through $t$ must be centred at $x$. Finally, from the construction, it is clear 
that $q$ lies on the diameter of  ${\tt MEC}_x$, thus proving that that ${\tt MEC}_x$ 
encloses $q$. 
\qed
\end{proof}

Now we will describe how to solve the {\tt QMEC} problem for a convex
polygon. Given a  query point $q$ we find (using Lemma\ \ref{llx}) the point $x$ on 
$M$ such that ${\tt MEC}_x$ encloses $q$. 

Observe (informally for now) that, for any fixed point $q$ inside $P$, the {\tt MEC}s 
that encloses $q$ are centered on a connected subtree $M^q$ of the medial axis $M
$. This observation is formally proved in Lemma\ \ref{l1} in the more general setting of 
simple polygons. Coupling this observation with Lemma\ \ref{obs1}, we can conclude 
that $c_q$ is the point on $M^q$ closest to the root $c$ of $M$. Therefore, we can 
locate $c_q$ by performing a binary
search on the path $x\sim c$. We find two consecutive nodes $v$ 
and its parent $v'$ on the path $x\sim c$ such that ${\tt MEC}_v$ encloses $q$, but 
${\tt MEC}_{v'}$ does not.  Since the path lies on a 
tree representing the medial axis $M$, we can use level-ancestor 
queries \cite{BF04} for this purpose. After computing $v$ and $v'$, 
the exact location of $c_q$ can be computed in $O(1)$ time. Thus, 
we have the following theorem:

\begin{theorem} \label{th-QMEC-convex}
A convex polygon $P$ with  $n$ vertices can be preprocessed in $O(n)$ time
and space such that, given any arbitrary query point $q \in P$, the largest circle 
containing 
$q$ inside $P$ 
can be reported in $O(\log n)$ time.
\end{theorem}

\vspace{-0.2in}
\section{{\tt QMEC} problem for simple polygons} \label{sec:simple-polygon-case}
Let $P$ be a simple polygon on $n$ vertices. Recall that the
{\it medial axis} $M$ of $P$ is defined to be the locus of  the centers of all
circles inside $P$ that touch the boundary of
$P$ in two or more points (see, e.g.,  \cite{CSW99}). While the medial axis of
a convex polygon consists only of straight line segments, the medial axis of 
a simple polygon may additionally contain parabolic arcs \cite{P1977}.

Our approach for solving the 
{\tt QMEC} problem uses the fact that 
$M$ is a geometric tree. Its {\it leaf nodes} 
correspond to the vertices of $P$, and the {\it internal nodes}
correspond to the points on $M$ such that the {\tt MEC}s  centered at each 
of those points touch three or more distinct points on the boundary 
of $P$. We  denote the set of internal nodes of $M$ as ${\tt N}$. An {\it edge} in 
$M$ is a path between two nodes that does not contain  any other node in its interior. Note that 
 a single edge  consists of one or more line segments or parabolic arcs.

For any point $x\in M$, we denote the maximal empty circle centered at $x$ in $P$ 
by ${\tt MEC}_x$.
A point $x \in M$, that is not a leaf, is said to be a {\it
valley} (resp., {\it peak}) if for a positive $\delta \to 0$, the {\tt MEC}s
centered at points in $M$ within  distance $\delta$ from $x$ are at
least as large as (resp., no larger than) ${\tt MEC}_x$ and at least 
one such {\tt MEC} is strictly larger (resp., smaller) than $
{\tt MEC}_x$. Note that a pair of parallel edges in $P$ may induce a pair of peaks or a pair of 
valleys. In such cases, we only pick one representative peak or valley and discard the 
other. We use $\Phi$ and $\Theta$ to denote the set of valleys and 
peaks, 
respectively.  For any $x \in \Phi$, it is easy to observe that ${\tt MEC}_x$ touches $P
$ in exactly two points diametrically opposed to each other. Otherwise, we can move along a direction to get smaller {\tt MEC}s. 
 As a consequence, a valley can only be in the interior of an edge.
Therefore, $\Phi \cap {\tt 
N} = 
\emptyset$. On the other hand, $\Theta \subseteq {\tt N}$.

We define a {\em mountain} to be a maximal subtree of 
$M$ that does not contain any valley point (except as its leaves). 
We partition $M$ by cutting the tree at all the valley 
points, and generate a set of mountains ${\tt M} = \{M_1,M_2,\ldots,M_{|{\tt M}|}\}$ 
(See Figure \ref{fig:MedialAxisPartitioning}(a)). 
\begin{observation} \label{o} ~\\
(i) Each valley point is the common leaf of exactly two mountains.\\
(ii) Each mountain has exactly one peak.\\
(iii) If a point $x$ moves from a valley point of a 
mountain towards its peak, the size of ${\tt MEC}_x$ increases monotonically.  
\end{observation}
\begin{proof}
Suppose $v$ is a valley point. We have noted earlier that $v$ can only be in 
the interior of an edge. Therefore, $v$ must be a common leaf between exactly two mountains.

For part (ii), assume that there are two peaks $p_1$ and $p_2$ in a mountain. Consider the path in $M$ from $p_1$ to $p_2$ (denoted $p_1 \sim p_2$). Consider the point $x^* = \arg \min_{x \in p_1 \sim p_2} {\tt MEC}_x$.  One can observe that $x^*$ will be a valley, implying that $p_1$ and $p_2$ cannot be in the same mountain.

 Consider a point $x$ that moves from a valley  of a 
mountain towards its peak. If the {\tt MEC}s don't increase monotonically, 
it is easy to see that a valley point will be encountered. This
implies that $x$ has moved into another mountain.\qed
\end{proof}
\begin{figure}[h]
\begin{minipage}[c]{0.5\textwidth}
\begin{center} %
\includegraphics[width=0.95\textwidth]{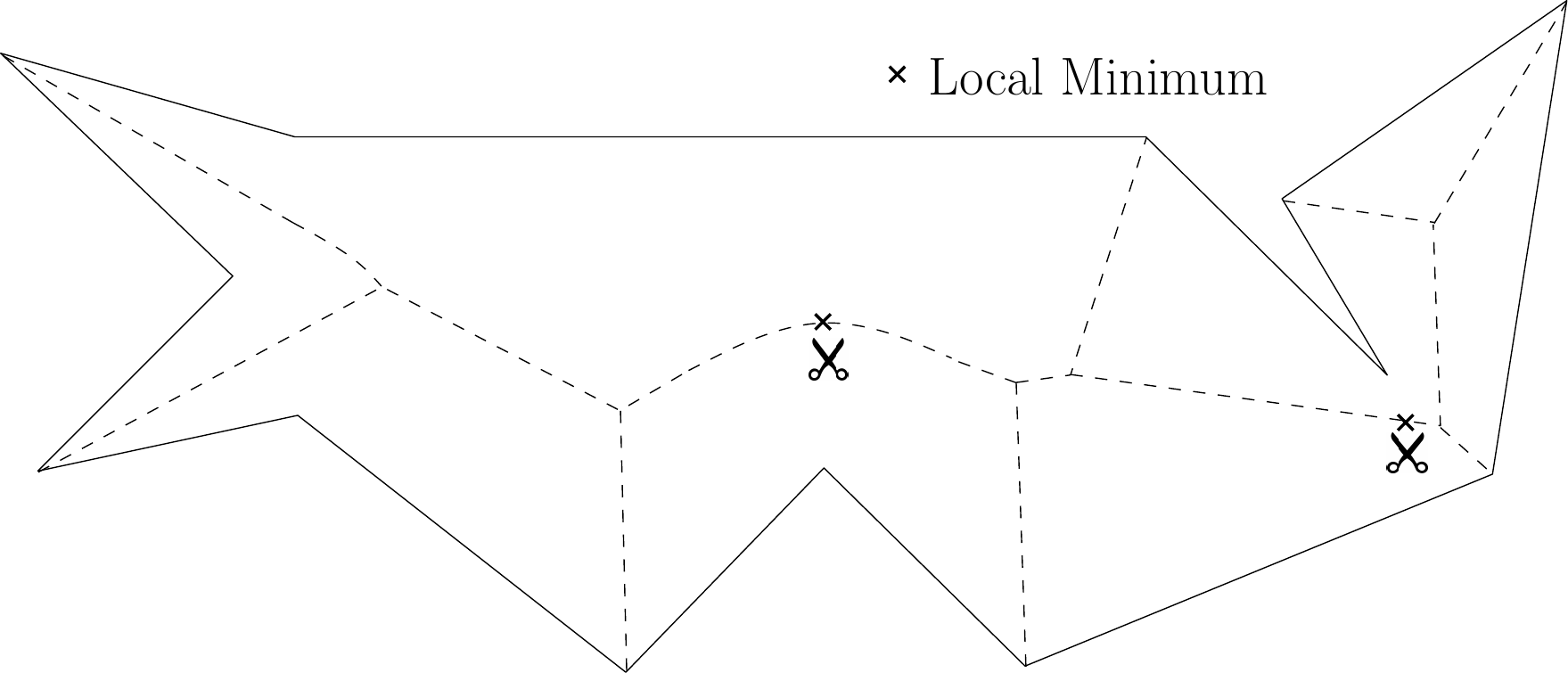}\\
(a)  
\end{center}
\end{minipage}%
\begin{minipage}[c]{0.5\textwidth}
\begin{center} %
\includegraphics[width=0.95\textwidth]{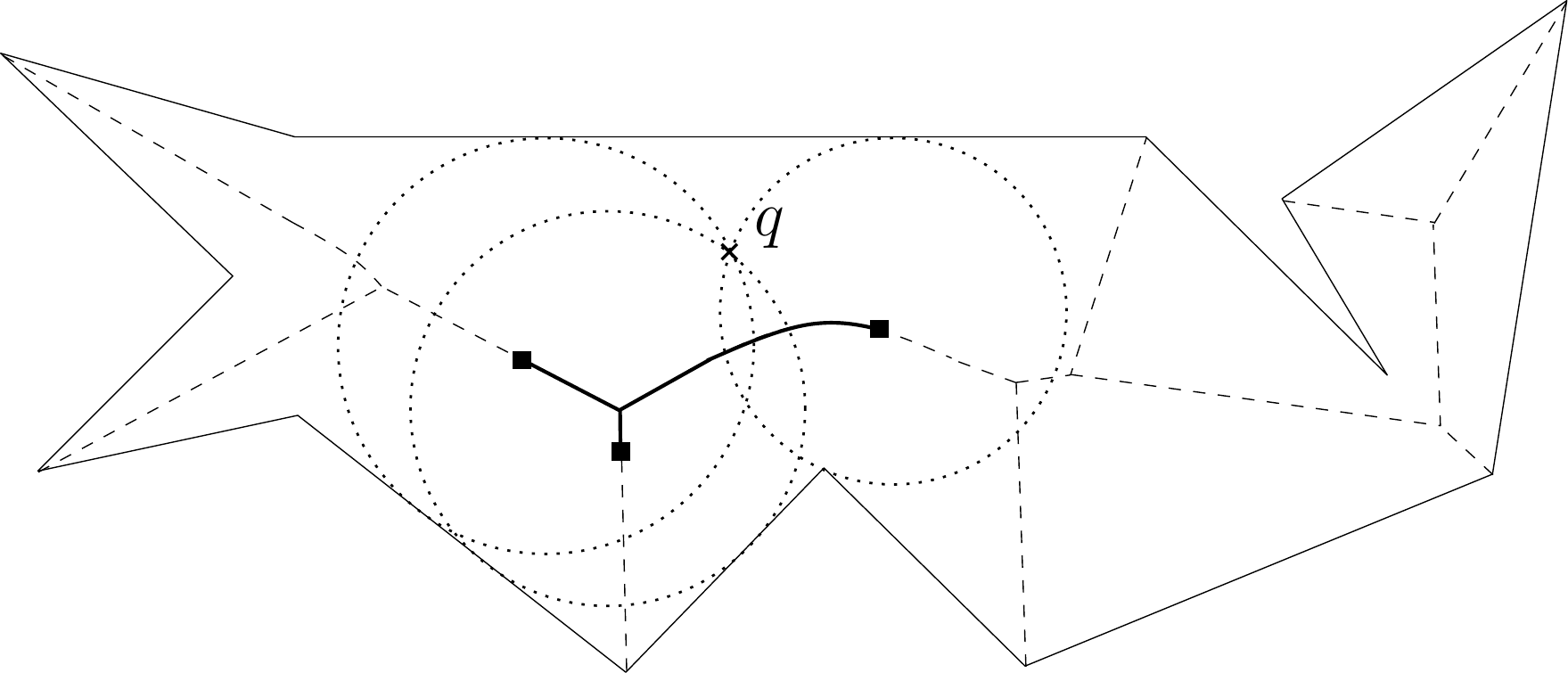}\\
(b)  
\end{center}
\end{minipage}
\caption{(a) Partitioning the medial axis 
$M$ into mountains, and (b) The subtree $M^q$ for a query point $q$}
\label{fig:MedialAxisPartitioning}
\end{figure}

At each valley point $x$ of $M$, consider the chord in $P$ connecting the two points 
at which ${\tt MEC}_x$ touches $P$. These chords induced by each $x \in \Phi$ will 
partition $P$ into  a set of sub-polygons $\{P_1,P_2, \ldots, P_{|{\tt M}|}\}$ of 
cardinality equaling the total number of mountains because this 
partitioning ensures that the portion of $M$  contained in each of these sub-polygons 
is a 
mountain containing a single peak. 

Given a (query) point $q\in P$, let $M^q\subseteq M$ denote the locus of the 
centers of all possible maximal empty circles in $P$ that enclose $q$ (see 
Figure \ref{fig:MedialAxisPartitioning}(b)). The following structural lemma 
plays a crucial role in designing our algorithm. 
\begin{lemma} \label{l1}
$M^q$ is a (connected) subtree of $M$.
\end{lemma}
\begin{proof}
For the sake of contradiction, let us assume that $M^q$ is disconnected. Then, there 
are at least two distinct  points $t_1$ and $t_2$ on the medial axis such that 
(a) ${\tt MEC}_{t_1}$ and ${\tt MEC}_{t_2}$ contain $q$, but (b) ${\tt MEC}_t$ 
does not contain $q$ for some $t$  on the path along the medial axis connecting 
$t_1$ and $t_2$.   

Without loss of generality, we assume that such a $t$ is not a node in $M$. 
Therefore, ${\tt MEC}_t$ touches the simple polygon $P$ at exactly two points, 
$a$ and $b$. The chord $[a,b]$ partitions $P$ into two polygons $P_{left}$ and 
$P_{right}$ to the left and right of $[a,b]$ respectively (see Figure~
\ref{fig:fig_lemma3}(a)). 
Note that, $t$ also partitions the medial axis into two subtrees, $M_{left}$ and 
$M_{right}$, 
such that $t_1 \in M_{left}$ and $t_2 \in M_{right}$. For the rest of the proof, 
we use ${\tt MEC}_t$ and $P_{left}$ to denote the region enclosed by  
them. We now claim that 

\begin{equation}\label{eqn:clm}
{\tt MEC}_{t_1} \subset P_{left} \cup {\tt MEC}_t = P \setminus (P_{right} \setminus {\tt 
MEC}_t).
\end{equation}

\begin{figure}[ht]
\begin{minipage}[b]{0.57\linewidth}
\centering
\includegraphics[width=0.7\linewidth]{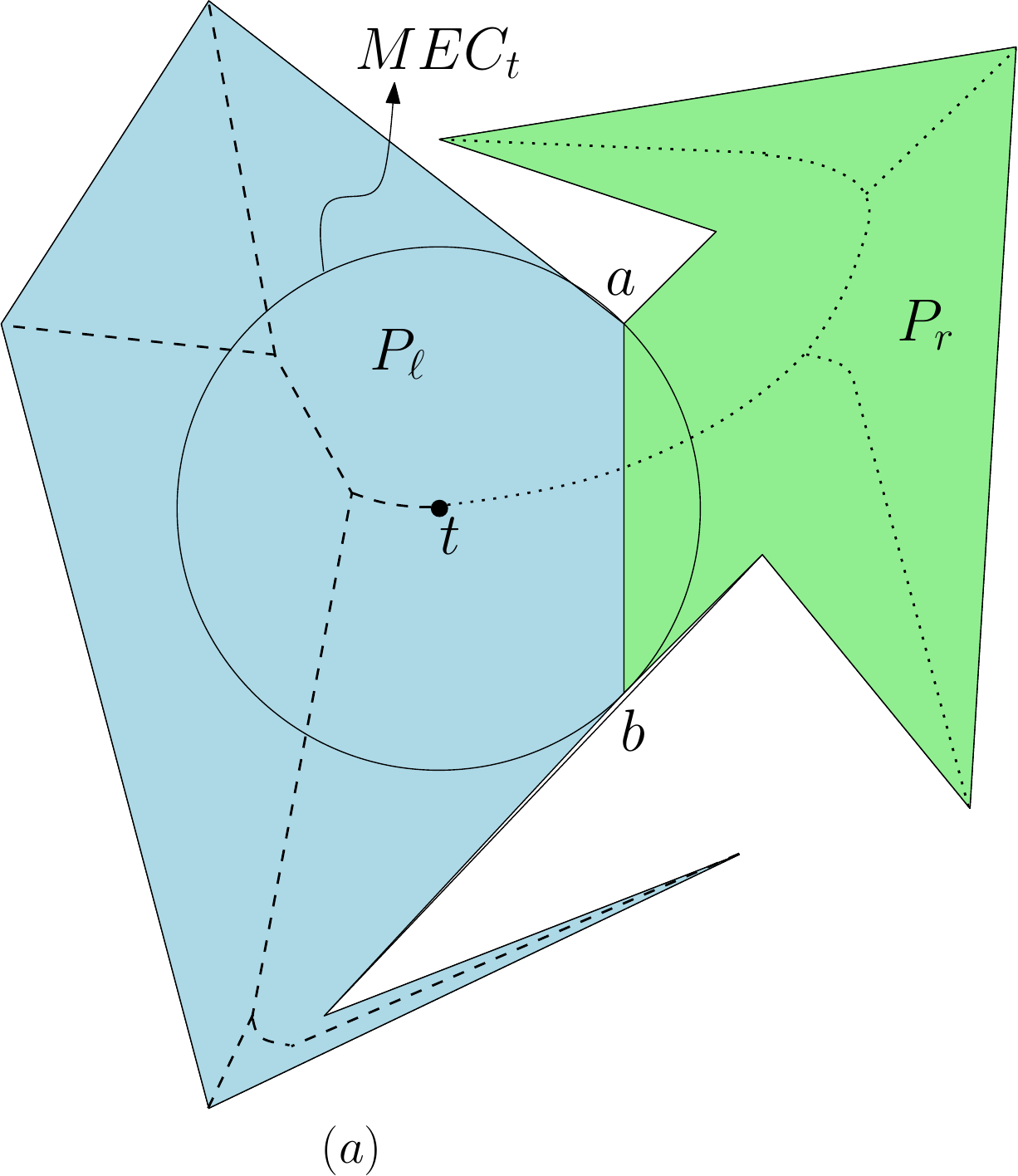}
\end{minipage}
\hspace{0.5cm}
\begin{minipage}[b]{0.37\linewidth}
\centering
\includegraphics[width=\linewidth]{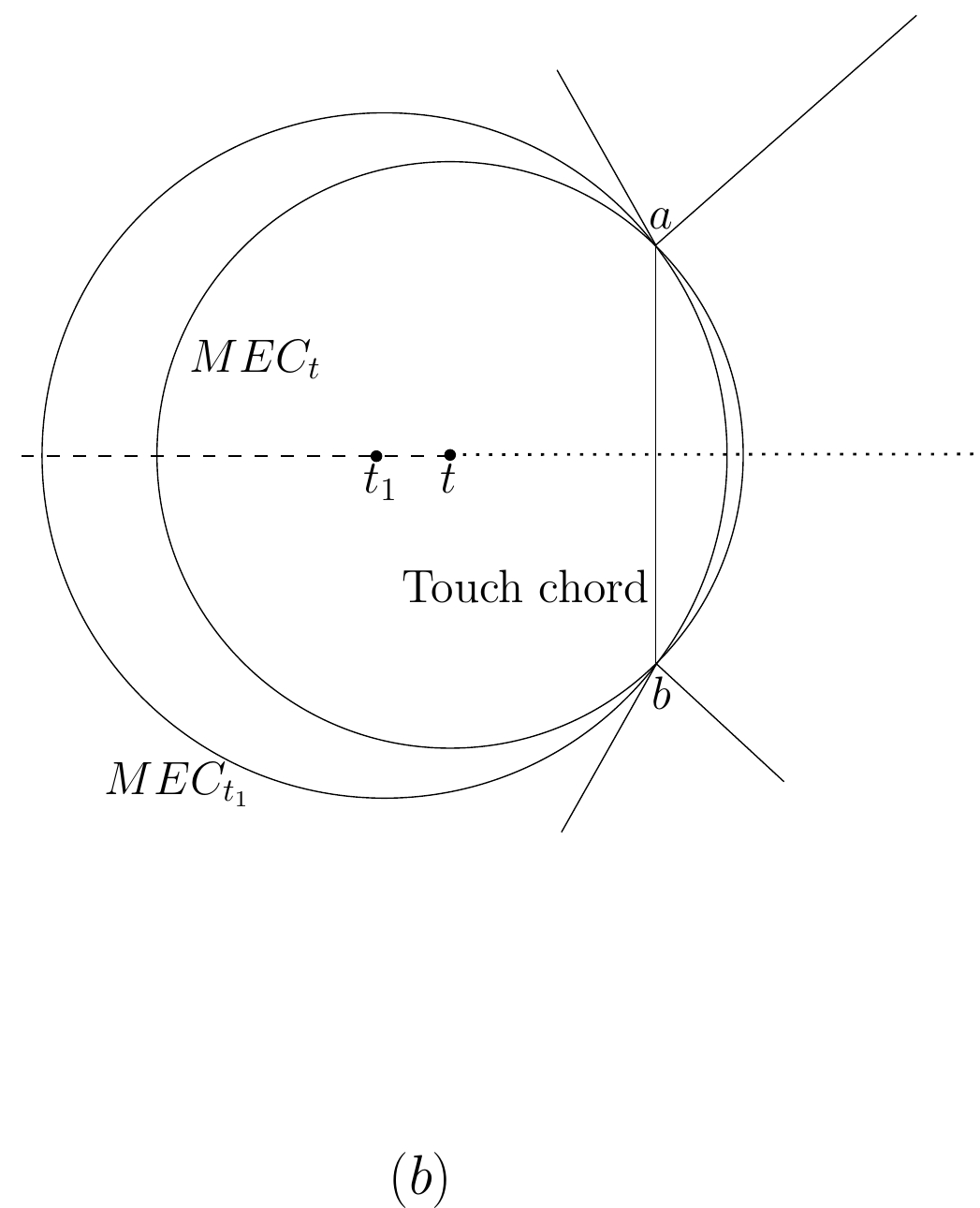}
\end{minipage}
\caption{Proof of Lemma \ref{l1}: (a) the general case, and (b) a special 
case where the points $a$ and $b$ are concave vertices of $P$}
\label{fig:fig_lemma3}
\end{figure}

To show that Equation \ref{eqn:clm} holds, we consider the following cases.

\begin{description}
\item[Case: ${\tt MEC}_{t_1}$ touches $P$ at both $a$ and $b$:] 
This case is illustrated in Figure~\ref{fig:fig_lemma3}(b). Here, $a$ and $b$ 
must be concave vertices that induce a straight line segment edge in the medial 
axis. Both $t$ and $t_1$ are on that edge; in particular, $t_1$ will be to the 
left of $t$. It is now easy to infer that Equation~\ref{eqn:clm} holds.

\item[Case: ${\tt MEC}_{t_1}$ touches at most one of $\{a,b\}$:]
Let ${\tt MEC}_{t_1}$ touch the other point $d \not\in \{a,b\}$ on the boundary of $P$. 
Clearly, $d \in P_{left} \setminus {\tt MEC}_t$. If we assume that  
${\tt MEC}_{t_1}$ also passes through a point $d' \in P_{right} \setminus {\tt MEC}_t$, 
then 
it is  
impossible to construct ${\tt MEC}_{t_1}$ without properly enclosing 
some point outside $P$. Therefore, Equation~\ref{eqn:clm} holds.
\end{description}
By symmetry, we can also say that 

\begin{equation}\label{eqn:sym}
{\tt MEC}_{t_2} \subset P_{right} \cup {\tt MEC}_t = P \setminus ( P_{left} \setminus {\tt 
MEC}_t).
\end{equation}
Taking the intersection  of Equations~\ref{eqn:clm} and \ref{eqn:sym}, we get

${\tt MEC}_{t_1} \cap {\tt MEC}_{t_2} \subset ( P_{left} \cup {\tt MEC}_t) \cap 
(P_{right} 
\cup {\tt MEC}_t )$ \\
= $(P_{left} \cap P_{right}) \cup {\tt MEC}_t$ 
= ${\tt MEC}_t$ (since  $(P_{left} \cap P_{right}) \subset {\tt MEC}_t$.

This contradicts our assumption that $q$ falls in ${\tt MEC}_{t_1}$ and ${\tt MEC}
_{t_2}$, 
but not in ${\tt MEC}_t$.
\qed 
\end{proof}

\begin{corollary}\label{l2}
Let $v\in M$ be such that ${\tt MEC}_v$ does not contain $q$. Then $M^q$ is 
contained 
entirely in one 
of the subtrees of $M$ obtained by deleting $v$ from $M$.
\end{corollary}
\begin{corollary}\label{cor:intermediate}
Consider two points $u, v \in M$, such that ${\tt MEC}_v$ overlaps with  ${\tt MEC}_u
$.
Let $\alpha$ be a point in $P$ such that $\alpha \in {\tt MEC}_v \cap {\tt MEC}_u$.
Then $\alpha$ will lie in all {\tt MEC}s centered  along 
the path from $v$ to $u$ in $M$. 
\end{corollary}
\begin{proof}
Since ${\tt MEC}_v$ overlaps with  ${\tt MEC}_u$ and $\alpha \in {\tt MEC}_v \cap {\tt 
MEC}_u$, both $v$ and $u$ are points on $M^\alpha$. The result  now follows 
immediately  from Lemma\ \ref{l1}. 
\qed
\end{proof}

Before we delve into solving {\tt QMEC}, in the next three subsections, we define 
three data structures that we use as building blocks.

\vspace{-0.15in}
\subsection{{\tt PLiCA}: Point location in circular arrangement} \label{sec:plica}
The problem is to preprocess a set 
${\tt C} = \{C_1, C_2, \ldots, C_n\}$ of circles 
of arbitrary radii, so that for any query point $q$ in the plane, we need to 
quickly report if there exists a circle $C_i\in {\tt C}$ such that $q\in C_i$. 
This can be achieved by using the concept of Voronoi diagrams in Laguerre geometry 
of circles in $\tt C$ \cite{IIM}. Each cell of the Voronoi 
diagram is a convex polygon and is associated with a circle in 
${\tt C}$. The membership query is answered by performing a 
point location in the associated planar subdivision. 
The preprocessing time and space complexities are $O(n\log n)$ and 
$O(n)$ respectively, and the queries can be answered in $O(\log n)$ time.

\vspace{-0.15in}
\subsection{{\tt QiM}: Query-in-Mountain}\label{subsec:MiM}
Given a mountain $M_i \in {\tt M}$, we must preprocess it such that given a query point $q$ inside $P$ such that $M_i 
\cap 
M^q \neq \emptyset$ (and a point $x \in M_i 
\cap 
M^q$),
our task is to report the largest {\tt MEC}  centered at a point on $M_i$ that contains 
$q$. 
Note that if the center moves 
from $x$ towards the peak of $M_i$, the size of the 
{\tt MEC} increases monotonically. Thus, we can apply the algorithm proposed in 
Section\ \ref{sec:Convex} for the convex polygon case to identify the largest {\tt MEC}  
containing $q$, and centered on $M_i \cap M^q$. 
The preprocessing time and space complexities for the mountain $M_i$ are both 
$O(|P_i|)$, and the query time is $O(\log |P_i|)$, where $|P_i|$ denotes the 
number of edges in the sub-polygon  $P_i$ that induces $M_i$. Since the 
set of mountains and sub-polygons are partitions of the medial axis $M$ and the 
polygon $P$, respectively, all the 
mountains can be preprocessed for the {\tt QiM} queries in $O(n)$ time.

\subsection{{\tt QiC}: Query-in-Circle (Problem Definition and Bounds)} 
\label{qic_defn}
The {\tt QiC} is a simplification of the {\tt QMEC} problem in which, in addition to $P$ and its medial
axis $M$, a node $v$ of  $M$  is 
specified as part of the input for preprocessing. We are  promised that 
the query point $q$ will lie inside ${\tt MEC}_v$. As in ${\tt QMEC}$, we are to report 
the largest {\tt MEC} $C_q$ that contains $q$. We defer the details of our solution for the 
{\tt QiC} problem to Section\ \ref{subsec:QiC}, where we prove the following theorem.
\begin{theorem}\label{thm:qic}
There exists a solution for {\tt QiC} that takes $O(n \log n)$ time and $O(n)$ space for preprocessing.
Queries can be answered in $O(\log n)$ time.
\end{theorem}
To solve {\tt QMEC},  we employ a divide-and-conquer strategy that 
divides the medial axis into smaller pieces. On these smaller pieces, we employ {\tt 
QiC}. We remark in Section\ \ref{subsec:QiC} how the solution to the  {\tt QiC} 
problem on the entire medial axis can be adapted to restricted portions of the medial 
axis. For now, we note that the preprocessing time and space of {\tt QiC} scale with the size of the portion of the medial axis that is preprocessed.
On a  
portion $M^*$ of the medial axis, the preprocessing time and space are $O(n^* \log n^*)$ time and $O(n^*)$, respectively, where $n^*
$ is the number of edges of $P$ that induce\footnote{We say that an edge $e_P$ in 
$P$ induces an edge  $e_M$ in $M$ if for some point $x$ in the interior of $e_M$, ${\tt MEC}_x$ 
touches $e_P$.} the edges in $M^*$ (cf. Corollary\ \ref{cor:mstar}).

\begin{algorithm}[h!]
\caption{Preprocessing for {\tt QMEC} on a Simple Polygon $P$}
\label{alg:qmecSP}
\begin{algorithmic}[1]
\REQUIRE{A simple polygon $P$.}
\STATE Compute the medial axis $M$ of $P$.

\STATE Construct a {\tt PLiCA} data structure on {\tt MEC}s centered on nodes of $M
$.
\STATE Construct a secondary {\tt PliCA} data structure on the {\tt MEC}s centered on 
valley points of $M$.
\STATE Construct a list $(M_1, M_2, \ldots )$ of mountains and preprocess each 
mountain for {\tt QiM}.
\STATE Partition $P$ into sub-polygons $P_1$, $P_2$, $\ldots$ such that each $P_i$ 
is associated with its corresponding $M_i$. (Recall that this can be performed by 
cutting along diameters of {\tt MEC}s centered on valley points.) 
\STATE Preprocess  $P$ and its sub-polygons (in $O(n)$ time and 
space) such that, given a query point $q$, the sub-polygon that contains $q$ can be 
reported efficiently (in $O(\log n)$ time). Call this data structure $D$. \label{lno:easy}
\STATE $T \leftarrow$ Decompose($M$). \label{lno:decompose}

\FOR{each node $t \in T$} \label{lno:for}
\STATE Let $M^t \subseteq M$ be the subtree associated with $t$.
\STATE Let $v^t$ be the centroid of $M^t$ (cf. Lemma\ \ref{Jordon}).
\STATE Preprocess ${\tt MEC}_{v^t}$ for {\tt QiC} with the additional promise that 
the largest empty circle $C_q$ that contains query point $q$ is centered on 
$M^t$. Associate this {\tt QiC} data structure with $t$. \label{lno:qic}
\ENDFOR \label{lno:endfor}

\end{algorithmic}
\end{algorithm}

\begin{algorithm}[h!]
\caption{Decompose($M$)}
\label{alg:decompose}
\begin{algorithmic}[1]
\REQUIRE{A tree $M$ with  $n$ nodes.}
\ENSURE{A divide-and-conquer tree $T$ that decomposes $M$.}
\IF{$n=1$}
\RETURN a tree with the single node.
\ENDIF
\STATE Find the centroid (cf. Lemma\ \ref{Jordon}) of $M$ that will decompose $M$ 
into subtrees $M^1$, $M^2$, $\ldots$
\STATE Create a  node {\tt node} with a list {\tt child} of child pointers.
\STATE Associate $M$ with {\tt node}.
\FOR{each subtree $M^i$}
\STATE  {\tt node.child[i]} $\leftarrow$ Decompose($M^i$).
\ENDFOR
\end{algorithmic}
\end{algorithm}

\begin{algorithm}[h!]
\caption{Query phase of {\tt QMEC} on a simple polygon $P$ with query point $q$}
\label{alg:QMECquery}
\begin{algorithmic}[1]
\REQUIRE{Query point $q$ and all the data structures created in the preprocessing 
phase.}

\COMMENT{We can use the {\tt PLiCA} data structure for the following condition.}
\IF{$q$ falls inside some {\tt MEC} centered on a node $v$ of $M$}
\STATE \COMMENT{We are in the affirmative case.}
\STATE Find the node $t$ in $T$ whose centroid is $v$.\\
\COMMENT{In the next step, $v^*$ is the centroid of the subtree of $M$ associated 
with $t^*$.}
\STATE Find (in $T$) the farthest ancestor $t^*$ of $t$  such that ${\tt MEC}_{v^*}$ 
contains $q$.
\STATE $C_q \leftarrow $ circle returned by querying the {\tt QiC} data structure 
associated with $v^*$. 
\STATE Return $C_q$. 
\ELSE
\STATE \COMMENT{We are in the negative case.}
\IF{$q$ falls in some  {\tt MEC} centered on a valley point $p$}
\STATE \COMMENT{$p$ is a valley point connecting exactly two mountains $A$ and 
$B$.}
\STATE $C_q$ is the larger of the {\tt MEC}s returned by querying the {\tt QiM} data 
structures associated with $A$ and $B$.
\STATE Return $C_q$.
\ELSE 
\STATE Use data structure $D$ to find the sub-polygon $P^q$ that contains $q$.
\STATE $C_q \leftarrow $ circle returned by querying the {\tt QiM} data structure 
associated with $P^q$.
\STATE Return $C_q$
\ENDIF
\ENDIF
\end{algorithmic}
\end{algorithm}

\vspace{-0.2in}
\subsection{Preprocessing for the {\tt QMEC} problem} \label{preprocessing}

Algorithm\ \ref{alg:qmecSP} outlines the steps in the preprocessing phase. The first 
\ref{lno:easy} steps are straightforward. Before we explain the subsequent steps, we 
state (for the sake of completeness) a well-known lemma.
\begin{lemma} \label{Jordon} \cite{jordon}
Every tree $M$ with $n$ nodes has at least one node $v$ whose 
removal splits the tree into subtrees with at most 
$\lceil\frac{n}{2}\rceil$ nodes. The node $v$ is called the 
centroid of $M$.
\end{lemma}

In line number\ \ref{lno:decompose} we call Algorithm\ \ref{alg:decompose} 
(recursively) to build  a centroid decomposition tree $T$. We partition   using the 
centroid (cf. Lemma \ref{Jordon}) in order to ensure that $T$ is balanced. 

The centroid decomposition is constructed in anticipation of the query phase. Suppose $q$ is a query point. If $q$ lies in  ${\tt MEC}_v$, where $v$ is  the centroid  associated with the root of ${\tt T}$. Then, we can use the ${\tt QiC}$ attached  to the root (in line number\ \ref{lno:qic}). If, on the contrary, $q \not \in {\tt MEC}_v$, then 
from Corollary\ \ref{l2} we know that only one of the 
subtrees rooted at $v$ will contain $M^q$, thereby allowing us to recurse into that 
subtree until we find the centroid whose {\tt MEC} encloses $q$. To facilitate this recursion, we must provide a way to find the correct subtree 
to recurse into. For this, we consider the geometry of the polygon $P$. Let ${\tt MEC}
_v$ touch the polygon $P$ at $k$ 
($\geq 3$) 
points. Consider the partitioning of  $P$ into $k$ sub-polygons, apart 
from the 
one containing $v$, by inserting chords as shown in Figure \ref{recursion}. These 
$k$ sub-polygons correspond to the $k$ subtrees of $M$ obtained by removing $v$. It is easy to 
see now that a point location data structure will suffice. In the query phase, 
we can simply find the sub-polygon that contains $q$ and recurse into the 
corresponding subtree.

\begin{figure}
\centering
\includegraphics[width=0.7\columnwidth]{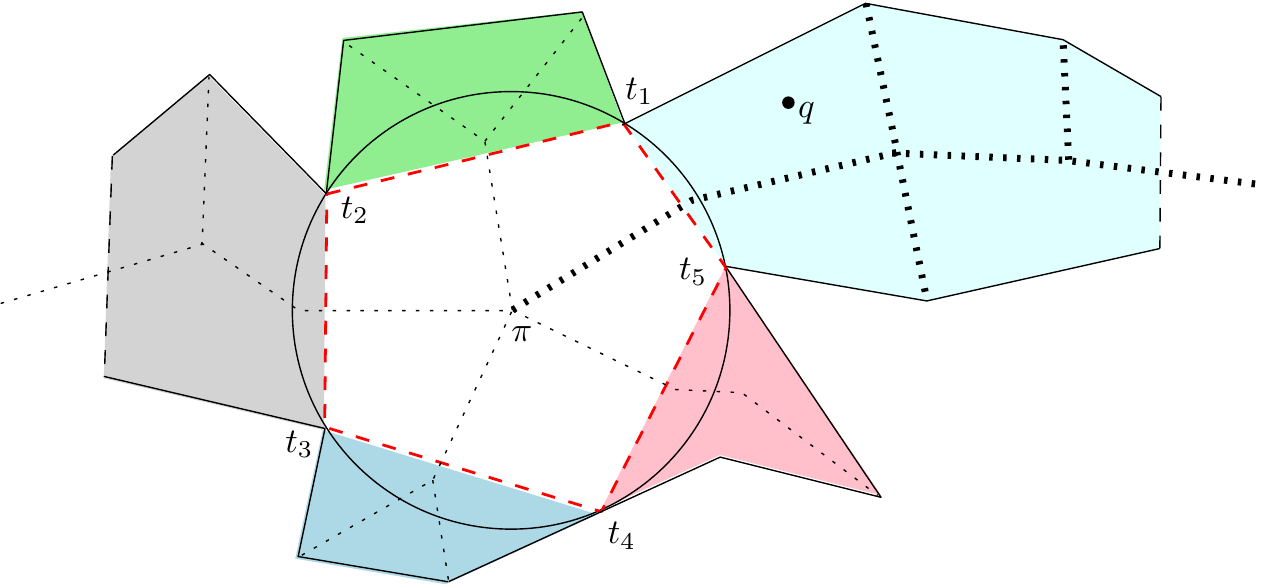}
\caption{The divide and conquer search structure}
\label{recursion}
\end{figure}

In lines \ref{lno:for} to \ref{lno:endfor}, for each node $t$ of $T$ we
associate an appropriate subtree $M^t$ of $M$ along with the centroid $v^t$ of $M^t
$. Additionally, we will construct a {\tt QiC} data structure associated with $t$ with the 
additional promise that the largest empty circle $C_q$ that contains query point $q$ is 
centered on $M^t$.

\begin{lemma} \label{sv}
The time and space required for preprocessing $P$ 
are $O(n\log^2 n)$ and $O(n\log n)$, 
respectively.
\end{lemma}
\begin{proof}
The medial axis $M$ of a simple polygon can  be computed in $O(n)$ time 
\cite{CSW99}. Once we have $M$, the partition of $M$ into mountains, and the 
associated 
partitioning of
$P$  can be done in $O(n)$ time. 
The data structure for the planar point location can easily be obtained in $O(n\log n)$ 
time 
and $O(n)$ space. The {\tt PLiCA} 
data structure for all the {\tt MEC}s centered at the nodes of $M$ requires $O(n\log n)
$ 
time and $O(n)$ space. 

Consider a level $\ell$ in the tree $T$. Each node in level $\ell$ implements the {\tt 
QiC} data structure on a  portion of the medial axis that is disjoint from the portion 
addressed by other {\tt QiC} implementations in the same level. Therefore,  the 
preprocessing times and spaces of all {\tt QiC}s at any level $\ell$ is $O(n \log n)$ 
and $O(n)$ respectively. Since there are $O(\log n)$ levels in $T$, the total 
preprocessing time and space for all {\tt QiC}s is $O(n \log^2 n)$ and $O(n \log n)$, 
respectively.
\qed
\end{proof}

\vspace{-0.2in}
\subsection{{\tt QMEC} query}
As discussed in Algorithm\ \ref{alg:QMECquery}, in the query phase with a point $q$, 
we first test whether $q$ lies in the {\tt MEC} centered at any node $v$ of 
the medial axis. This can be performed using the {\tt PLiCA} data structure 
(see Subsection \ref{sec:plica}) built over the set of {\tt MEC}s centered 
at all the nodes in $M$. We now need to consider two cases:

{\bf Affirmative case:} There exists a node $v$ in $M$ such that ${\tt MEC}_v$ 
contains $q$. From $t \in T$ corresponding to $v$ we move upward in the centroid  
tree $T$ following the parent 
pointers to identify a node $t^* \in T$ at the highest level such that the ${\tt MEC}$ 
centered at ${v^*}$, 
the medial axis node of the subtree associated with $t^*$, 
contains $q$. Our choice of $v^*$ coupled with Corollary\ \ref{l2} imply the following 
lemma.
\begin{lemma}\label{lem:promise}
Let $M^*$ be the subtree of $M$ associated with $t^* \in T$ and let $v^*$ be the 
centroid of $M^*$. 
Then, $M^q \subseteq M^*$ and ${\tt MEC}_{v^*}$ contains $q$.
\end{lemma}
Lemma\ \ref{lem:promise}  ensures all the prerequisites for {\tt QiC}, so we can query 
the {\tt QiC} data structure associated with  node $t^*$ and correctly obtain $C_q$.

{\bf Negative case:} There exists no node $v$ in $M$ such that ${\tt MEC}_v$ 
contains $q$. In this case, $M^q$ cannot span more than two  mountains as 
otherwise $M^q$ must include a node in $M$. If $q$ falls in an {\tt MEC} centered at a valley point $p$, we 
query the {\tt QiM} data structure associated with the two polygons connected by $p$. 
Otherwise, 
the center of $C_q$ lies in only one mountain, the mountain that 
contains $q$. We identify the sub-polygon $P_i$  in the planar subdivision that 
contains $q$ using the data structure {\tt D} (see line number\ \ref{lno:easy} of 
Algorithm\ \ref{alg:qmecSP}). Finally, we can compute 
$C_q$ by performing the {\tt QiM} query in $M_i$, the mountain associated with $P_i
$.

\begin{theorem} \label{th:th}
Given a simple polygon $P$, we can preprocess it in $O(n\log^2 n)$ time and 
$O(n\log n)$ space, such that for a query point $q\in P$, the largest circle 
$C_q$ in $P$, that contains $q$, can be reported in $O(\log n)$ time. 
\end{theorem}
\begin{proof}
The correctness follows from the above discussion.
Preprocessing time and space have already been established in Lemma \ref{sv}.
We now analyze the  query time. The {\tt PLiCA} query requires $O(\log n)$ time 
\cite{IIM}. If we are in the affirmative case, then finding the node $v^*$ 
at the maximum level in $T$ such that $q \in {\tt MEC}_{v^*}$ needs another 
$O(\log n)$ time. The {\tt QiC} query for ${\tt MEC}_{v^*}$ can be executed in $O(\log n)$ 
time (Lemma \ref{lem:prep}). In the negative case, finding the appropriate sub-
polygon 
and then performing the ${\tt QiM}$ query requires $O(\log n)$ time.\qed
\end{proof}

\vspace{-0.15in}
\subsection{Description of the {\tt QiC} Data Structure}\label{subsec:QiC}
Recall from Section\ \ref{qic_defn} that  the {\tt QiC} data structure preprocesses  the 
medial axis and a specified node $v$ such that when queried with a point $q \in {\tt 
MEC}_v$, the largest {\tt MEC} containing $q$ can be reported efficiently.  

We use the concept of {\it guiding circles} associated with node $v\in M$ to find $C_q
$.
Let ${\tt R}$ be an array containing the radii of the {\tt MEC}s centered at all the 
nodes 
in $M$, sorted in increasing order.

\begin{definition} (Guiding Circles of a node $v$ of $M$) \label{def1}
An {\tt MEC} $C$ centered somewhere on $M$ is called a guiding {\tt MEC} of the 
node 
$v$ of $M$ if (i) its radius is in ${\tt R}$, (ii) every {\tt MEC}  on the path from 
$v$ to the center of $C$ in $M$ (both inclusive) is no larger than $C$, and (iii) $C$ 
overlaps with ${\tt MEC}_v$.\footnote{We say that two circles overlap if they have a 
common point in their interior.} (See Figure\ \ref{fig:gc} for an illustration of guiding 
circles on a single path from $v$.) We denote the set of guiding circles of the node $v \in M$ by ${\tt S}_v$.
\end{definition}
\vspace{-0.2in}
\subsubsection{Computing ${\tt S}_v$.} \label{GC}
We can compute ${\tt S}_v$ by adapting  either depth-first search 
or breadth-first search 
traversal on $M$ starting from $v$. As we traverse $M$ using (say) depth first search 
starting from $v$, we keep track of the largest MEC along the path from $v$ to the 
current position in the traversal. When we encounter an {\tt MEC} $C$  that fits our 
definition of a guiding circle, we include $C$ in ${\tt S}_v$ along with the id of the 
mountain in which it is centered.

\begin{algorithm}[h!]
\caption{Preprocessing Phase of {\tt QiC}}
\label{alg:QiCPrep}
\begin{algorithmic}[1]
\REQUIRE {Polygon $P$, its medial axis $M$ and a vertex $v$ of $M$.}
\STATE Compute the radii of {\tt MEC}s centered at nodes of $M$ and store them in a sorted array $
{\tt R}_v$.
\STATE Compute the guiding circles ${\tt S}_v$ of node $v$. \COMMENT{We can use 
an adaptation of either depth first search or breadth first search.}
\STATE To each $C \in {\tt S_v}$ centered at $c$, attach  the mountain id of the mountain containing $c$.
\STATE For each $r \in {\tt R}_v$, attach the set $S_v^r \triangleq \{ s \in S_v |
\text{ radius of $s$ is $r$\}}$. \COMMENT{In Lemma\ 
\ref{lem:BoundingS}, we will see that, for any $r$, $|S_v^r|$ is a constant.}
\end{algorithmic}
\end{algorithm}

Before we provide the pseudocode for the query phase, we establish a few lemmas. 
Recall from Lemma\ \ref{cor:intermediate} that if a guiding circle $C$ contains $q$, 
then every guiding circle from ${\tt MEC}_v$ to $C$ will contain $q$. The proof for the 
following lemma follows from the definition of guiding circles.
\begin{lemma}\label{increasing} 
Let $\Pi$ be the path on $M$ from $v$ to some guiding circle $C$. 
\begin{enumerate}
\item The radii of guiding circles along $\Pi$ are non-decreasing. 
\item Furthermore, if $r \in {\tt R}$ is the radius of $C$ and $r_v$ is the 
radius of ${\tt MEC}_v$,  then for every  $r' \in {\tt R}$ such that $r_v \le r' \le 
r$, there is at least one guiding circle of radius $r'$ in the path from $v$ to the center 
of $C$. 
\end{enumerate}
\end{lemma}
\begin{corollary}\label{cor:binary}
Given a query point $q$, let  $\rho_q$ be the radius  in ${\tt R}_v$ such that $\exists 
C \in {\tt S}_v^{\rho_q}$ that contains $q$ but $\forall  r > \rho_q$, $ \nexists C \in {\tt 
S}_v^{r}$ that contains $q$. Then, for every $r' \in {\tt R}_v$ such that $r' \le \rho_q$, 
$\exists C \in {\tt S}_v^{r'}$ that contains $q$.
\end{corollary}

\begin{figure}
\centering
\includegraphics[width=0.9\columnwidth]{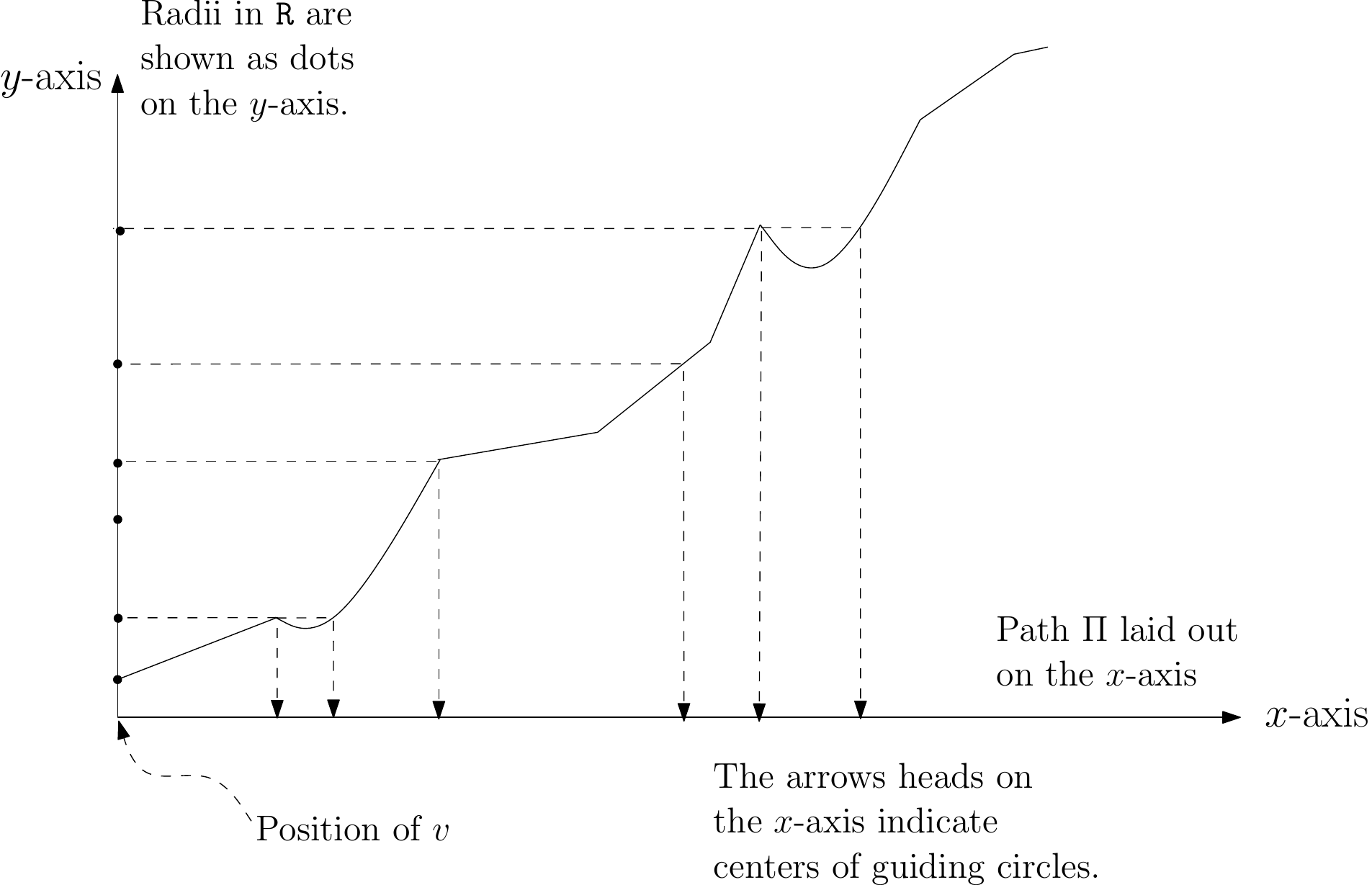}
\caption{For the sake of intuition on the construction and usefulness of guiding circles, 
we show a path $\Pi$ laid out on the $x$-axis. }
\label{fig:gc}
\end{figure}

Corollary\ \ref{cor:binary} will allow us to perform a binary search  for $\rho_q$ which 
in turn will lead us to the largest guiding circles in ${\tt S}_v$ that contain $q$. The following lemma ensures that  the binary search will run in $O(\log 
n)$ time.

\begin{lemma}\label{lem:BoundingS} For any $r \in {\tt R}_v$, 
$|{\tt S}_v^r|$ is bounded by a constant. 
\end{lemma}
\begin{figure}[t]
\vspace{-0.1in}
\begin{minipage}[c]{0.35\textwidth}
\begin{center} %
\includegraphics[height=1.75in]{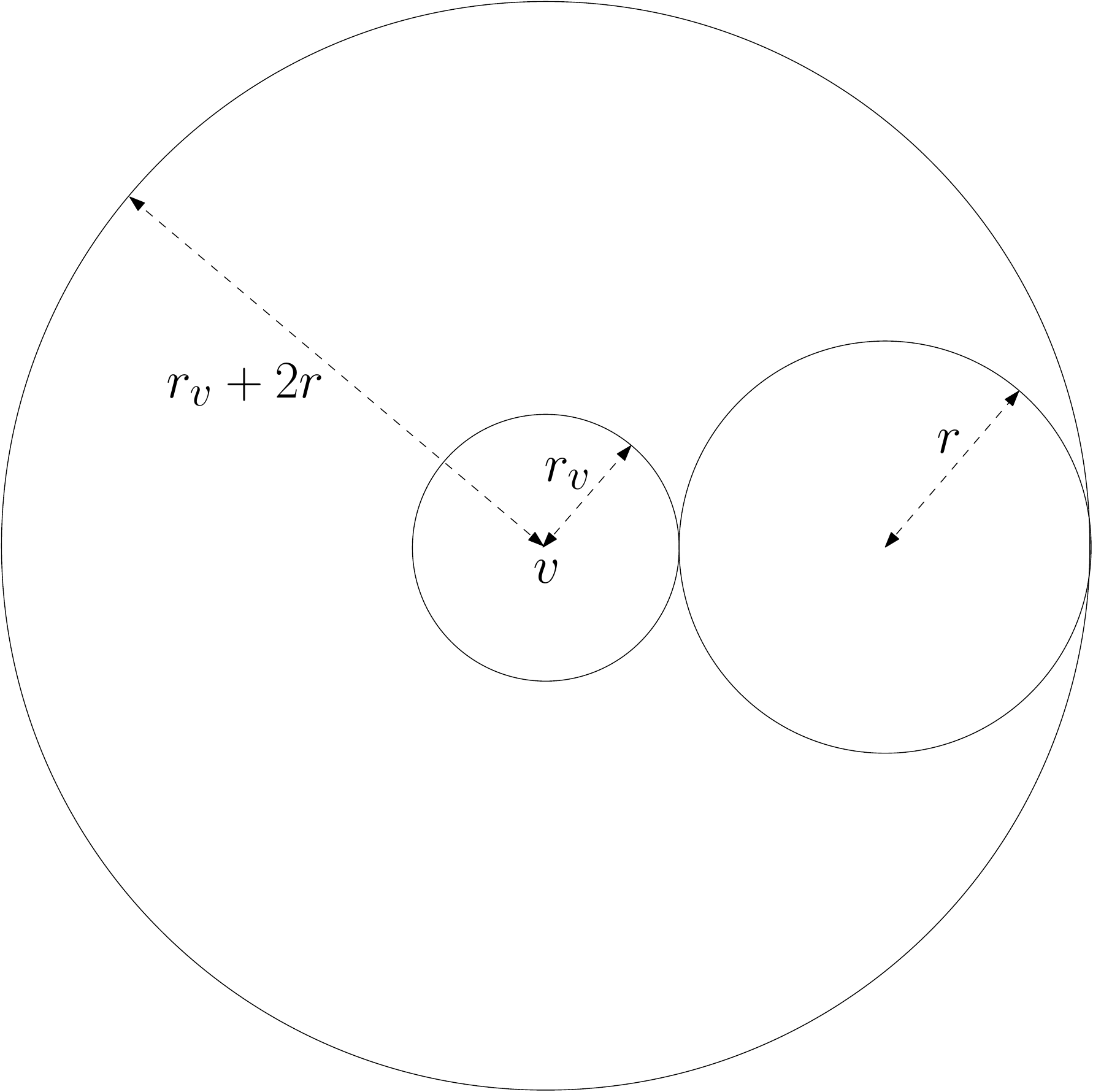}\\
(a)  
\end{center}
\end{minipage}%
\begin{minipage}[c]{0.65\textwidth}
\begin{center} %
\includegraphics[width=\textwidth] {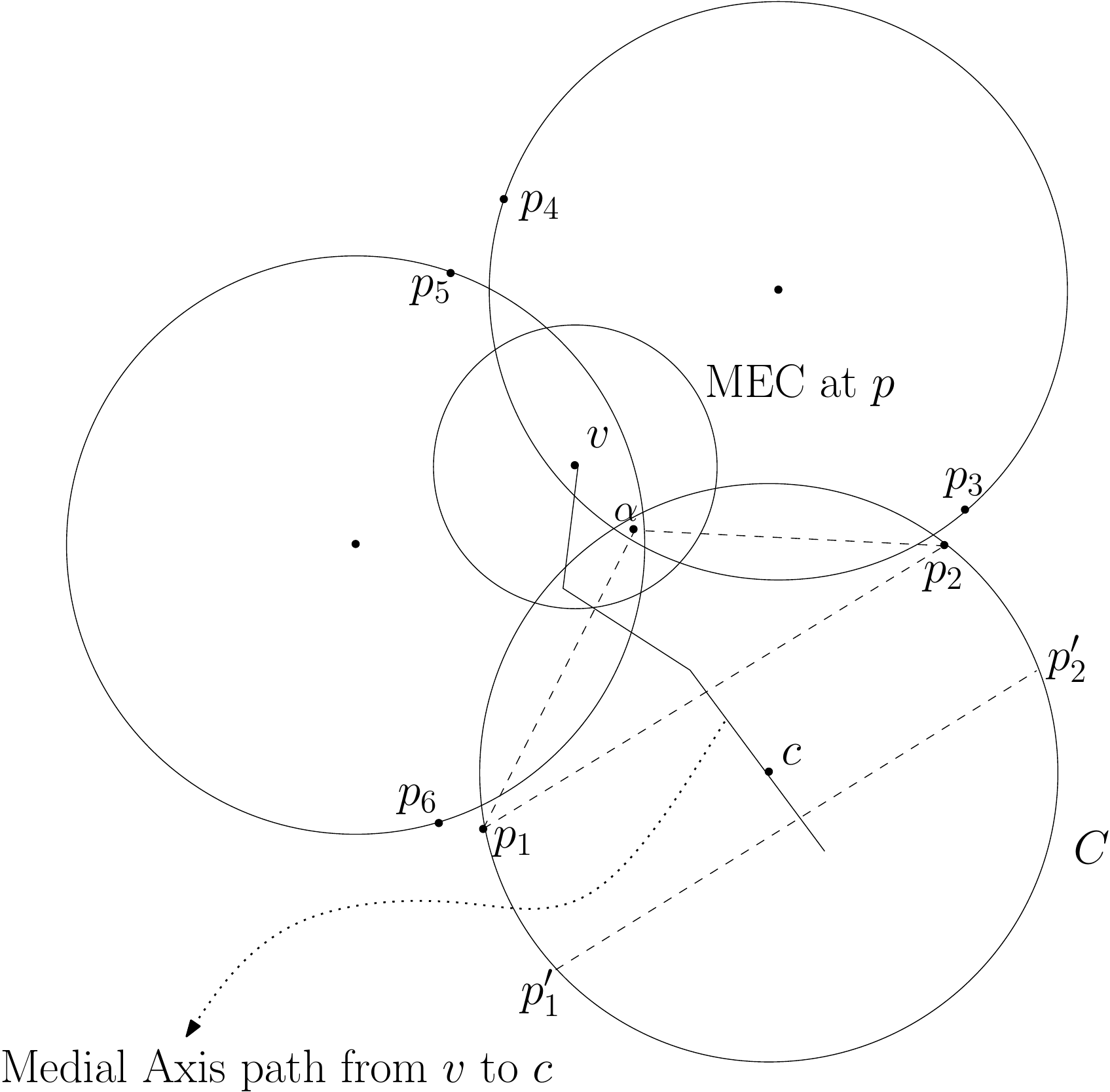}\\
(b)  
\end{center}
\end{minipage}
\vspace{-0.1in}
\caption{(a) Bounding $|{\tt S}|$, and (b) Illustration of Lemma \ref{lem:BoundingS}}
\vspace{-0.1in}
\label{fig:BoundingS_r}
\end{figure}

\begin{proof}
Consider any $r \in {\tt R}$. Let ${\tt S}_v^r$ be the {\tt MEC}s of radius $r$ in 
${\tt S}_v$. 
For convenience, let us assume that ${\tt S}_v^r$ does not contain {\tt MEC}  
centered at any node of $M$. Since {\tt MEC}s centered at nodes of $M$ have 
distinct radii, at most one {\tt MEC}  in ${\tt S}_v^r$ can be centered at a node.

By the condition (ii) of Definition \ref{def1}, $\rho_v \le r$. Also recall that every 
{\tt MEC} in ${\tt S}_v$ must intersect ${\tt MEC}_v$ (see 
Figure \ref{fig:BoundingS_r}(a)). Therefore, every {\tt MEC}  in ${\tt S}_v^r$ must lie 
entirely within a circle $\chi$ of radius $\rho_v + 2r$ centered at $v$. Thus, we need 
to prove that the number of guiding circles of radius $r$ at node $v$ inside $\chi$ is
bounded by a constant.

Let us consider a point $\alpha \in P$. Let ${\tt S}_v^r(\alpha) \subseteq {\tt S}_v^r$ 
be the set of {\tt MEC}s in $S_v^r$ that enclose $\alpha$. Consider any {\tt MEC}  $C 
\in 
{\tt S}_v^r(\alpha)$; let $c$ be its center. Let $p_1$ and $p_2$ be the two points at 
which $C$ touches the boundary of the polygon $P$. The chord $p_1p_2$ must 
intersect 
the medial axis (see Figure \ref{fig:BoundingS_r}(b)). Note that, the points $v$ and $c
$ 
lie on different sides of $p_1p_2$. On the contrary, if $v$ and $c$ lie in 
the same side of $p_1'p_2'$, where $p_1'$ and $p_2'$ are the points of contact of the 
said {\tt MEC}  and the polygon $P$, then we can increase the size of the {\tt MEC} by 
moving its center $c$ towards $v$ along the medial axis (see Figure 
\ref{fig:BoundingS_r}
(b)). 
Thus, $C \not\in {\tt S}_v^r$. Thus, we have $\angle p_1 \alpha p_2 \ge \pi/2$. Again, 
the 
angles subtended by different {\tt MEC}s in ${\tt S}_v^r(\alpha)$ are disjoint. These 
two 
facts imply that $|{\tt S}_v^r(u)| \le 4$. In other words, any point inside the circle 
$\chi$ can be enclosed by at most four different circles of ${\tt S}_v^r$. We need to 
compute $|{\tt S}_v^r|$. Let us consider a function $\eta(\alpha)$ defined as the 
number 
of circles in ${\tt S}_v^r$ that overlap at the point $\alpha$, $\alpha \in \chi$. Clearly, 
$\eta(\alpha) \leq 4$ for all $\alpha \in \chi$. The total number of circles in 
${\tt S}_v^r$ can be bounded as follows: 

Total area of circles in ${\tt S}_v^r \leq   \int_{\alpha=(x,y) \in \chi} \eta(\alpha) ~dx
~dy \le 4 \pi(r_v + 2r)^2$. \\ 

Therefore, $|{\tt S}_v^r|  \leq  \frac{4 \pi(r_v + 2r)^2}{ \pi r^2} \leq
\frac{4 \pi (3r)^2}{ \pi r^2} = 36$. \qed
\end{proof}

\vspace{-0.15in}
\subsubsection{Answering {\tt QiC} query}
Given a query point $q$ and a node $v$ in $M$ such that $q \in {\tt MEC}_v$, we 
compute 
$C_q$ as follows. Let $\rho \in {\tt R}$ be the radius of the largest guiding circle in
${\tt S}_v$ containing $q$, and ${\tt S}_v^\rho(q)$ be all the members of ${\tt S}_v$ 
with radius $\rho$ that contain $q$.  We report $C_q$ by executing the steps in Algorithm\ \ref{alg:QiCQuery}.

\begin{algorithm}[h!]
\caption{Query Phase of {\tt QiC}}
\label{alg:QiCQuery}
\begin{algorithmic}[1]
\REQUIRE {A query point $q$ lying inside ${\tt MEC}_v$.}

 \COMMENT{We want to find $\rho \in {\tt R}_v$ such that $\exists C \in {\tt S}_v^\rho
$ that contains $q$ but $\forall  r > \rho$, $ \not \exists C \in {\tt S}_v^{r}$ that contains 
$q$.}

\STATE Perform a binary search in the array ${\tt R}_v$ to identify 
$\rho \in {\tt R}_v$. This also returns the members in ${\tt S}_v^\rho(q)$. 
Note that, each member $C \in {\tt S}_v^\rho(q)$ is attached with its 
corresponding mountain-id.

\STATE For each member in $C \in {\tt S}_v^\rho(q)$, locate the largest 
{\tt MEC}  containing $q$ in the mountain attached to $C$ by executing the 
{\tt QiM} query algorithm.

\STATE Report $C_q$ as the largest one among the {\tt MEC}s obtained in
Step 2. 
\end{algorithmic}
\end{algorithm}

Lemmata \ref{lem:correctness} and  \ref{lem:prep} state the correctness and 
complexity. 

\begin{lemma}\label{lem:correctness}
At least one of the circles in ${\tt S}_v^\rho(q)$ is centered at some point 
on the mountain in which $C_q$ is centered. 
\end{lemma}

\begin{proof}
Since $M^q$ is a subtree of $M$ (Lemma \ref{l1}), if we explore all the paths in $M$ 
from node $v$ towards its leaves, the center $c_q$ of $C_q$ is reached in one 
of these paths, say $\Pi$. Let $C'$ be the  last guiding circle when going from $v$ to 
$c_q$. Let the center of $C'$ be $c'$. As a consequence of Lemma\ \ref{increasing}, 
$C' \in {\tt S}_v^\rho(q)$. Suppose for the sake of contradiction, $C'$ is not centered 
on the same mountain on which $c_q$ is centered. Then,  between $c'$ and $c_q$ 
there is a valley point 
$\alpha$ on the path $\Pi$, such that the radius of ${\tt MEC}_\alpha$ is 
less than the radius of ${\tt MEC}_{c'}$. Also, there exists another point 
$\beta$ on the path $\Pi$ between $\alpha$ and $c_q$ such that the radius of 
${\tt MEC}_\beta$ is equal to the radius of ${\tt MEC}_{v'}$. Since the radius 
of ${\tt MEC}_\beta$ matches with an element of ${\tt R}$, ${\tt MEC}_\beta$ 
must  also be a guiding circle. This contradicts our assumption that $C'$ is the last
guiding circle between $v$ and $c_q$.  
\qed
\end{proof}

\begin{lemma} \label{lem:prep}
(i) For a node $v$, ${\tt S}_v$ can be computed in $O(n \log n)$ time and 
$O(n)$ space. (ii) If the query point $q$ lies in ${\tt MEC}_v$, then $C_q$ can be 
computed in $O(\log n)$ time. 
\end{lemma}

\begin{proof}
(i) First of all, note that $|{\tt R}| \in O(n)$. The breadth first search in $M$ needs 
$O(n)$ time. The time for computing 
the members in ${\tt S}_v$ is $\sum_{r \in {\tt R}} |{\tt S}_v^r| = O(|{\tt R}|)$ 
(by Lemma \ref{lem:BoundingS}), which may be $O(n)$ in the worst case. 
A sorting of the members in ${\tt S}_v$ with respect to their radii is required; this takes 
$O(n \log n)$ time. Once sorted, attaching ${\tt S}_v^r$ with each $r \in {\tt R}$ will 
take $O(n)$ time. The space requirement can be argued similarly.

(ii)  If $q \in {\tt MEC}_v$, the binary search in ${\tt S}_v$ considers at most 
$O(\log |{\tt R}|) = O(\log n) $ distinct radii. For each radius, the number of guiding 
circles inspected to find whether any one contains $q$ is bounded by a constant 
(see Lemma \ref{lem:BoundingS}). Thus $\rho$, the largest radius among the guiding 
circles of node $v$ that contains $q$, can be identified in $O(\log n)$ time.

Let ${\tt S}_v^\rho(q)$ denote the set of guiding circles of node $v$ of radius 
$\rho$ that contains $q$. Each of them is attached with the corresponding mountain-
id. 
For each member in $C \in {\tt S}_v^\rho(q)$, we invoke {\tt QiM} query to find the 
largest {\tt MEC} in the associated mountain $M_i$; this takes $O(\log |M_i|)$ time, 
where $M_i$ may be $O(n)$ in the worst case (see Subsection \ref{subsec:MiM}).  
\qed
\end{proof}

\begin{corollary} \label{cor:mstar}
Suppose {\tt QiC} is restricted to a connected $M^* \subseteq M$, i.e., $v \in M^*$ 
and $M^q \subseteq M^*$. Suppose further that $n^*$ edges of $P$ induce the edges 
in $M$. Then, the preprocessing time and space for {\tt QiC} are $O(n^* \log n^*)$ and 
$O(n^*)$, respectively. The query time will be $O(\log n^*)$.
\end{corollary}
\begin{proof}
We can restrict our ${\tt R}$ to radii of {\tt MEC}s centered on nodes only in $M^*$. 
Hence $|{\tt R}| \in O(n^*)$. 
Rest of the proof follows from the previous discussion.
\qed
\end{proof}

\begin{proof}[of Theorem\ \ref{thm:qic}]
The proof follows from Lemma\ \ref{lem:correctness} and Lemma\ \ref{lem:prep}.
\qed
\end{proof}

\section{{\tt QMEC} problem for Point Set}
\label{sec:QMEC}
The input consists of a set of points $P=\{p_1, p_2, \ldots, p_n\}$ in 
$\mathbb{R}^2$. The objective is to preprocess $P$ such that given any 
arbitrary query point $q \in \mathbb{R}^2$, the largest circle $C_q$ that
does not contain any point of $P$ but contains $q$, can be reported 
quickly. Observe that, if $q$ does not lie in the interior of the convex 
hull of $P$, then we can easily report a circle of infinite radius passing 
through $q$, that does not overlap with $P$. So, in the rest of this section, 
we shall consider the case where $q$ lies in the  interior of the convex 
hull of $P$. 

Consider the Voronoi diagram of $P$. Observe that the {\tt MEC} centered at 
any {\it Voronoi vertex} touches at least three points of $P$. To simplify 
our presentation, we assume  that {\tt MEC}s 
centered at Voronoi vertices are of distinct sizes. 
In the course of our algorithm, we  treat the Voronoi diagram of $P$, as a 
plane graph  $G$. To keep $G$ within a finite region, we  insert artificial vertices,
one for each unbounded edge in the Voronoi diagram 
of $P$,  so that $G$ is the plane graph of the Voronoi diagram of $P$ with each 
unbounded edge clipped at 
its corresponding artificial vertex. In placing the artificial vertices, we ensure that 
(i) every {\tt MEC} centered at an artificial vertex 
must be larger than all the {\tt MEC}s centered at Voronoi vertices, and 
(ii) 
the {\tt MEC}s centered at artificial vertices do not overlap pairwise within 
the convex hull of $P$. They may overlap outside the convex hull of $P$.
The second condition ensures that there exists no query point $q$, in the convex hull 
of $P
$,
which can be enclosed by more than one {\tt MEC} centered at artificial vertices.
From now onwards, we will use the term {\em vertices of $G$} to collectively refer 
to Voronoi vertices and artificial vertices. We will use both the geometric and 
graph theoretic properties of $G$. In particular, to achieve the subquadratic 
preprocessing time, we use the classical planar separator theorem\ \cite{LT79}. The intuition is as follows.

Consider the following naive approach to solving {\tt QMEC} on points.
Suppose we store the {\tt MEC}s of vertices in $G$ in a {\tt PLiCA} data structure. Suppose, furthermore, that we preprocess each vertex for {\tt QiC} adapted for points set. We note that
here also the {\tt QiC} data structure of a vertex can be implemented using 
guiding circles (cf. Section\ \ref{sec:QiC-1} for details).
Given an arbitrary query point $q$ in 
the convex hull of $P$, we know that $q$ lies in at least one of the {\tt MEC}s centered on vertices\footnote{The {\tt MEC}s on vertices can be thought of as circumcircles of triangles in the Delaunay triangulation of $P$ and therefore the union of these {\tt MEC}s covers the entire convex hull region.}, 
so we can locate one such {\tt MEC}, say ${\tt MEC}_v$ for some vertex $v$. We can
execute the {\it {\tt QiC} query for the point set} to identify $C_q$.  This, unfortunately, will require $O(n^2 \log n)$ time for preprocessing because each {\tt QiC} preprocessing requires $O(n \log n)$ time. Instead, to achieve sub-quadratic bounds on the preprocessing time and space for the {\tt 
QMEC} problem we employ a divide-and-conquer approach
by recursively splitting the vertices of $G$ using the planar separator theorem 
 stated below.

\begin{theorem}\label{Planar-Separator} \cite{LT79}
A planar graph $G$ on $n$ vertices can, in $O(n)$ time, be partitioned into disjoint 
vertex sets $A$, 
$B$, and $W$ such that (i) $|W| \in O(\sqrt{n})$, (ii) $|A|,|B| \le 2n/3$, and (iii) 
there is no edge in $G$ that joins a vertex in $A$ to a vertex in $B$. 
\end{theorem}

\newcommand{\gin}{{G_{\text{in}}}}
\newcommand{\win}{{W_{\text{in}}}}
\begin{algorithm}[h!]
\caption{Preprocessing Phase of {\tt QMEC} for set $P$ of points in $\mathbb{R}^2$}
\label{alg:QMECPrep}
\begin{algorithmic}[1]
\REQUIRE {This is a recursive algorithm. In the first call, the input  graph $\gin$ is $G
$. Subsequently, the input graph $\gin$ is a subgraph of $G$.}
\ENSURE{The first call on the entire plane graph $G$ will return a pointer $r$ to 
the root of the separator decomposition tree ${\tt T}$.}
\IF{$\gin$ is empty}
\RETURN {\sc null}.
\ENDIF

\STATE Create a node $v$ of the separator decomposition tree ${\tt T}$.

\STATE Compute the planar separator vertices $\win$ of $\gin$ (cf. Theorem\ 
\ref{Planar-Separator}). Denote the two separated subgraphs as $A$ and $B$.
\STATE Compute a {\tt PLiCA} data structure $\Phi$ on the {\tt MEC}s centered on 
vertices in $\win$. 
\STATE Compute a {\tt PLiCA} data structure $\Theta$ on the {\tt MEC}s centered on 
vertices in $A$.
\STATE Compute a {\tt QiC} data structure corresponding to node $v$. \COMMENT{This {\tt QiC} data structure is built on the two promises that (i) {\it at least} one {\tt MEC} centered on the planar separator vertices $\win$ will enclose the query point $q$, and (ii) $C_q$ is centered on an edge of the plane graph $\gin$.}

\STATE Attach $\Phi$, $\Theta$ and the {\tt QiC} data structures to $v$. 

\STATE $v.\textsc{LeftChild} \leftarrow $ (Call Algorithm\ \ref{alg:QMECPrep} on $A$).
\STATE $v.\textsc{RightChild} \leftarrow $ (Call Algorithm\ \ref{alg:QMECPrep} on $B
$).

\RETURN pointer to $v$.

\end{algorithmic}
\end{algorithm}

\begin{algorithm}[h!]
\caption{Query Phase of {\tt QMEC} for set $P$ of points in $\mathbb{R}^2$}
\label{alg:QMECQuery}
\begin{algorithmic}[1]
\REQUIRE {A query point $q$ inside the convex hull of $P$ and pointer $r$ to root of the separator 
decomposition tree ${\tt T}$.}
\ENSURE{The larges {\tt MEC} $C_q$ that contains $q$  is returned.}

\STATE {\sc ptr} $\leftarrow$ $r$.

\WHILE{{\sc ptr} is not {\sc null}}
\STATE Let $t \in {\tt T}$ be the node pointed by {\sc ptr}.
\STATE Let $G_t$ be the subgraph of $G$ associated with $t$.
\STATE Let $W_t$ be the separator vertices of $G$.
\STATE Let $A_t$ and $B_t$ be the two disconnected subgraphs obtained when vertices in $W_t$ are removed from $G_t$.
\STATE Let $\Phi_t$, $\Theta_t$, and ${\tt QiC}_t$ be the data structures attached to $t$.
\IF{$\exists w \in W_t$  such that $w$ 
encloses $q$ (we check this using $\Phi_t$)}
\STATE $C_q \leftarrow$ (circle returned by querying ${\tt QiC}_t$ with query point $q$).
\RETURN $C_q$.
\ENDIF
\IF{$\exists w'$ in the data structure $\Theta$ associated with $v$ such that $w'$ 
encloses $q$}
\STATE {\sc ptr} = $v.\textsc{LeftChild}$.
\ELSE
\STATE {\sc ptr} = $v.\textsc{RightChild}$.
\ENDIF
\ENDWHILE
\STATE \COMMENT{The execution will not reach this point.}

\end{algorithmic}
\end{algorithm}

We construct a separator decomposition tree ${\tt T}$ as follows (cf. Algorithm\ \ref{alg:QMECPrep} for a detailed pseudocode). The root $r$ of 
${\tt T}$ represents the plane graph $G$. We  attach two {\tt PLiCA} data structures 
$\Phi$ and $\Theta$ at $r$. In $\Phi$, we store {\tt MEC}s centered on the 
$O(\sqrt{n})$ planar separator vertices (denoted by $W$). We also build the {\tt 
QiC} data structure for the node $r$. The details of {\tt QiC} in the current context of a set of points (rather than polygon)  is described in Subsection \ref{sec:QiC-1}. For now, however, we state the {\tt QiC} problem in the current context where $P$ is a set of points. The node $r$ has $O(\sqrt{n})$ {\tt MEC}s corresponding to the $O(\sqrt{n})$ separator vertices and therefore, the {\tt QiC} attached to $r$ comes with two promises. The first promise is  that the query point $q$ will be enclosed by \textit{at least} one of the separator {\tt MEC}s. (Note that this first promise is an adaptation from the context where $P$ is a simple polygon. In that context, because the medial axis of $P$ was a tree, the separator was  a single vertex.) Our second promise is that $C_q$ is centered on some edge of the plane graph $G$ attached to $r$. In the query phase of {\tt QiC}, given a query point $q$, we are to return $C_q$.

The removal of the  vertices in $W$ from $G$ will induce two disjoint subgraphs $A$ and $B$. Without loss of generality, we pick $A$  and build $\Theta$ 
containing the {\tt MEC}s centered at the vertices of $A$. The root $r$ 
has two children $\textsc{LeftChild}$ and $\textsc{RightChild}$ in ${\tt T}$. $\textsc{LeftChild}$ is associated with the 
 subgraph $A$ while  $\textsc{RightChild}$ is associated 
with 
the  subgraph $B$. The two children of $v$ are then
processed recursively. 

In the query phase (cf. Algorithm\ \ref{alg:QMECQuery}), we are given a query point $q$. We find the highest node $t$ in {\tt T} such that (at least) one of the {\tt MEC}s $C$ stored in the associated $\Phi$ encloses $q$. We find $t$ by a traversal from the root node of {\tt T}. Let $v$ be the center of $C$. The point $v$ is a separator vertex in the graph $G_t$ associated with $t$. Recall that each separator vertex has a {\tt QiC} data structure (restricted to $G_t$) associated with it. We prove subsequently that when we query the {\tt QiC} data structure attached to $v$ with the query point $q$, we will indeed obtain the largest {\tt MEC} $C_q$ that encloses $q$. 

We now turn our attention to analyzing Algorithm\ \ref{alg:QMECPrep} and Algorithm\ \ref{alg:QMECQuery}. We begin with some important lemmata.

\begin{lemma}\label{lem:cycle} 
Consider any cycle $H$ in the Voronoi diagram of $P$. Let $C_H$ be an {\tt 
MEC} centered at some point on $H$. Then, there exists another {\tt MEC} 
$C'_H$ centered at some other point on $H$ that does not properly overlap 
with $C_H$. 
\end{lemma}
\begin{proof}
Clearly, any cycle in the Voronoi diagram of $P$ must contain at 
least one point from $P$ inside it. Let $p \in P$ be such a  point
that lies inside the cycle $H$ (see Figure~\ref{fig:nspa}). Let $C_H$
be any {\tt MEC} centered at some point on $H$; let $c_H$ be the center of $C_H$.
Consider the line connecting $c_H$ and $p$. It intersects $H$ at
another point $c'_H$. It is easy to see that the {\tt MEC} $C_H'$, centered
at $c'_H$, will not properly overlap $C_H$ as, otherwise,
$p$ will lie inside both $C_H$ and $C'_H$. 
\end{proof}
\begin{figure}[h]
\begin{center} 
\includegraphics[height=3in,clip=true,trim=0 50 0 50]{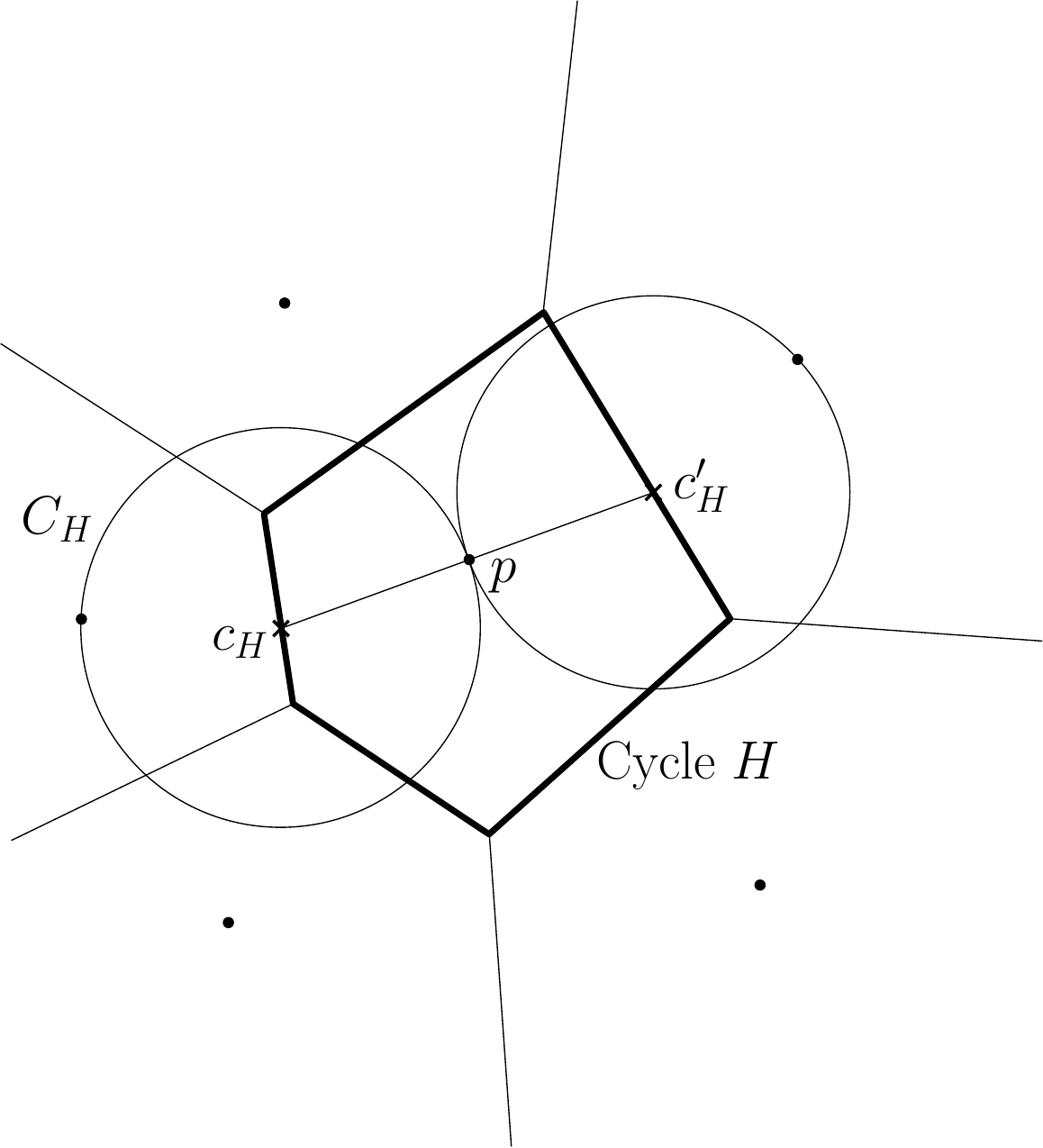}\\

\end{center}

\caption{Illustration for Lemma~\ref{lem:cycle}}
\label{fig:nspa}
\end{figure}

\begin{figure}[h]
\begin{center} 
\includegraphics[width=\textwidth,clip=true,trim=0 100 50 180]{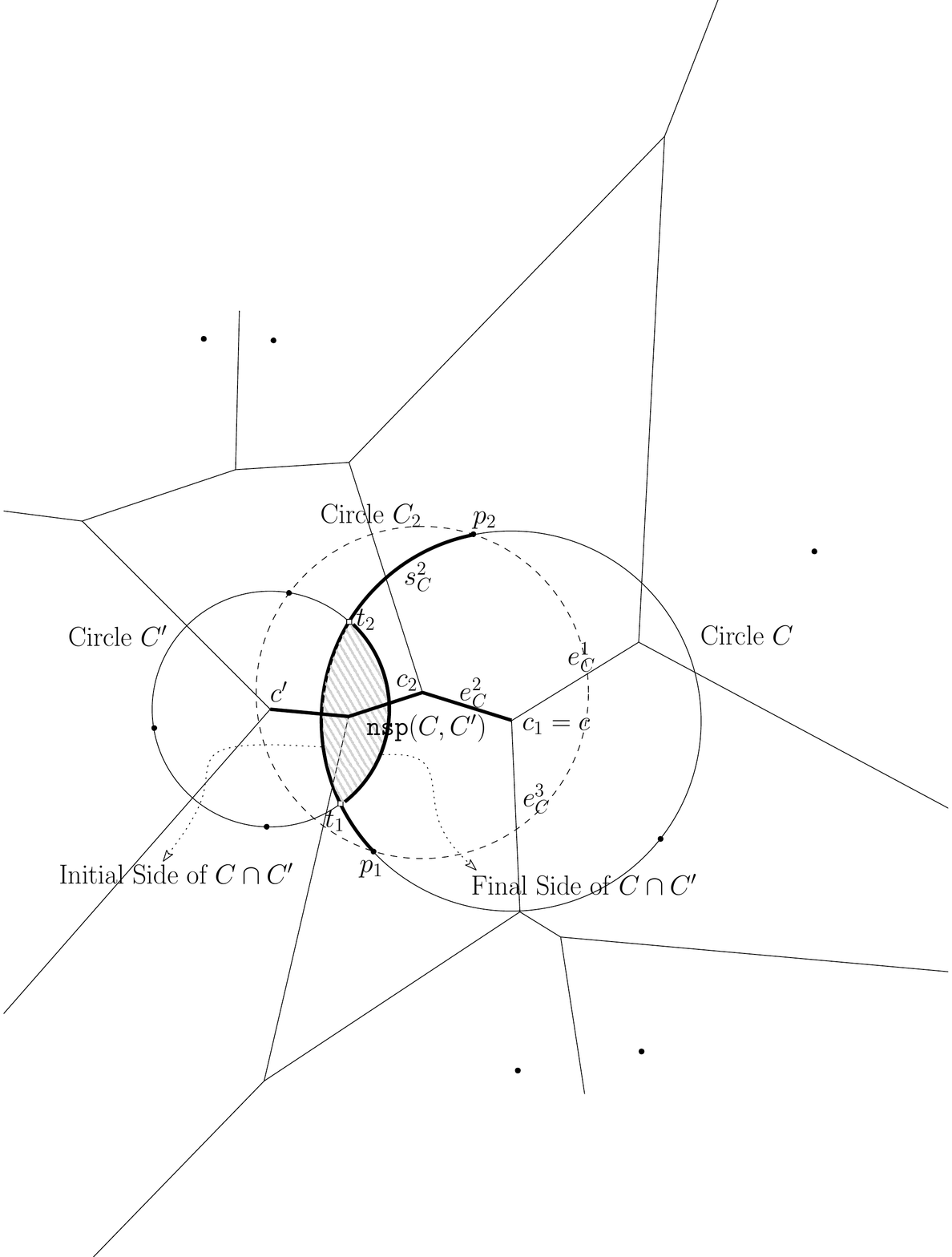}\\
 \end{center}
\caption{Illustration of $\Pi(c,c')$ in unique path lemma} 
\label{fig:nsp}
\end{figure}

\begin{lemma}(Unique Path Lemma.)
If $C$ and $C'$ are two distinct but overlapping {\tt MEC}s with centers
at $c$ and $c'$, respectively, then there is a unique path $\Pi(c,c')$ 
from $c$ to $c'$ in  the Voronoi diagram of $P$ such that every {\tt MEC} centered 
on that path encloses $C\cap C'$.
\end{lemma}

\begin{proof}
 The structure of the proof is as follows. We provide a procedure that 
constructs a path $\Pi(c,c')$ from $c$ to $c'$ along the Voronoi
edges, and ensure that every {\tt MEC} centered on that path encloses $C
\cap C'$. As a consequence of Lemma \ref{lem:cycle}, the path does not
form an intermediate cycle and terminates at $c'$. Finally, we again
use Lemma~\ref{lem:cycle} to show that no path $\Pi'$, other than
$\Pi(c,c')$, exists between $c$ and $c'$ such that every {\tt MEC} centered
on $\Pi'$ contains $C \cap C'$. Throughout this proof, we closely follow
Figure~\ref{fig:nsp} in order to keep the arguments intuitive. To keep
arguments simple, we assume that $c$ and $c'$ are Voronoi vertices. The arguments 
hold 
even
when $c$ and $c'$ are not Voronoi vertices. 

Let $\alpha$ be the number of points in $P$ that $C$  touches. These
$\alpha$ points partition  $C$ into $\alpha$ arcs. The degree of
the corresponding Voronoi vertex $c$ (center of $C$) is also $\alpha$ because each
adjacent pair of points of $P$ on the boundary of $C$ will
induce a Voronoi edge incident on $c$ and vice verse. These Voronoi
edges and their corresponding arcs are denoted by $e^i_C$ and $s^i_C$,
for $1 \le i \le \alpha$.

Consider the other  {\tt MEC} $C'$ ($\neq C$ and centered at a vertex $c'$)
that overlaps with $C$. $C'$ intersects $C$ at two points $t_1$ and
$t_2$. Since $C'$ is empty, both $t_1$ and $t_2$ must lie on one of the $\alpha$ arcs 
of
$C$. Let us name this arc by $s^j_C$.
Consider the edge $e^j_C = (c,c_2)$ that corresponds to the arc
$s^j_C$. The other end of $e^j_C$, i.e., the vertex $c_2$, is called
the {\em next step from $c$ toward $c'$} and denote it as
$\mathbf{ns}(c,c')$. Consider the pseudocode in Procedure~\ref{alg:pi} that generates 
the 
path
denoted by $\Pi(c,c')$:

\begin{algorithm}[h!]
\caption{$\Pi(c,c')$ Computation}
\label{alg:pi}
\begin{algorithmic}[1]
\STATE $\Pi(c,c') \leftarrow (c)$
\STATE $\mathbf{next} \leftarrow c$
\REPEAT
\STATE $\mathbf{next} \leftarrow \mathbf{ns}(\mathbf{next},c')$
\STATE Append $\mathbf{next}$ to $\Pi(c,c')$
\UNTIL{$\mathbf{next}$ equals $c'$.}
\COMMENT{Note that this is the only terminating condition.}
\end{algorithmic}
\end{algorithm}

We now show that (i) $\Pi(c,c')$ is our desired path, and 
(ii) there exists no other path satisfying the {\it unique path lemma}.

{\bf Proof of correctness:} 
Algorithm\ \ref{alg:pi} constructs a path $\Pi(c,c') = (c_1=c, c_2, \ldots,
c_i, c_{i+1}, \ldots, c')$, where each $c_i$ is a vertex in the Voronoi diagram of $P$. Let $C_2$ denote the
{\tt MEC} centered at $c_2$. If $C_2=C'$, then the procedure
terminates and, as required, every {\tt MEC} centered on the edge $(c,c_2)$ 
encloses $C \cap C_2 = C \cap C'$.  

Therefore, we consider the case where $C_2 \neq C'$. We need to prove 
$C \cap C' \subseteq C \cap C_2$.

Let $p_1, p_2 \in P$ be the points at which $C$ and $C_2$ intersect; $p_1,p_2$ are
the end points of the arc $s^j_C$ that defines the move of the next step 
toward $c'$ (in Figure \ref{fig:nsp}, $j$ is $2$). By
definition, $t_1$ and $t_2$ lie on the arc $s^j_C$. Notice that $C
\cap C'$ (shaded region in Figure~\ref{fig:nsp}) is shaped like a rugby
ball with $t_1$ and $t_2$ at its end-points. One side of $C \cap C'$
(called the {\it initial side}) is in $C$ and the other side 
(called the {\it final side}) is in $C'$. Clearly, $t_1$ and
$t_2$ are inside (or on the boundary of) every {\tt MEC} centered on the
edge $e^j_C$. Otherwise, as we go from $C$ to $C_2$, a
circle would be present that must touch the final side of $C \cap C'$, but that would 
mean
that we have either 
\begin{itemize}
\item reached $C'$, which contradicts our assumption that $C_2 \neq C'$, 
\item or found a {\tt MEC} that contains $C'$, which contradicts the fact
that $C'$ is itself a {\tt MEC}.
\end{itemize}

We now make two observations: (i) $C$ touches the initial side of $C
\cap C'$, but (ii) no other {\tt MEC} centered on $e^j_C$ ($C_2$ in particular)
touches the final side of $C \cap C'$. 

Observation (i) is obvious. We prove Observation (ii) by contradiction. Let  $C^*$ be an {\tt MEC} centered on $e^j_C$ that touches the final side
of $C \cap C'$ at, say, some point $t^*$. It is easy to see that $C^*$
will contain $C'$ because $C^*$ touches $C'$ at the point $t^*$ and also $C^*$ 
contains 
$t_1$
and $t_2$, which are also on the boundary of $C'$. Thus we have a contradiction that 
$C'$ 
is an {\tt MEC}. 
Thus, it is clear that $C \cap C' \subseteq C \cap C_2$.  

Consider two adjacent vertices $c_i$ and $c_{i+1}$ along $\Pi(c,c')$
with {\tt MEC}s $C_i$ and $C_{i+1}$ centered on them, respectively. The
above argument can be easily extended to give us the following: 
$$C_i \cap C' \subset C_{i+1} \cap C'.$$ Therefore, we can conclude
that every {\tt MEC} along $\Pi(c,c')$ encloses $C \cap C'$. 
Lemma~\ref{lem:cycle} suggests that $\Pi(c,c')$ does not
form a cycle. The only stopping condition is when we actually reach
$c'$. So $\Pi(c,c')$ terminates at $c'$ in at most $O(n)$ steps. 

{\bf Proof of uniqueness:} To complete the proof of this lemma, we 
must show that $\Pi(c,c')$ is the only required path. For the sake 
of contradiction, assume that there is another path $\Pi'$ such that 
every {\tt MEC} centered on $\Pi'$ contains $C \cap C'$. Then, there are 
two distinct paths from $c$ to $c'$ such that every {\tt MEC} centered 
on both the paths contain $C \cap C'$. Clearly, there must be a cycle 
when the two paths are combined. From Lemma~\ref{lem:cycle}, we know 
that there are pairs of {\tt MEC}s in the cycle that do not overlap on 
each other. This is a contradiction. Thus $\Pi(c,c')$ is the only 
required path.  
\qed
\end{proof}

Recall that,
given a query point $q$, we locate $C_q$ by traversing the tree ${\tt T}$ from its root 
node $r$. At each node $t$ on the search path, we search in the $\Phi_t$ data 
structure 
to check whether $q$ lies in an {\tt MEC}s corresponding to a separator vertex of 
node $t$.
If there exists an ${\tt MEC}_v \in \Phi_t$ containing $q$, then we perform {\tt QiC} 
query in $t$ to identify $C_q$.   
Otherwise,  we search $q$ in $\Theta_t$, 
associated with the partition $A_t$. Now, if there exists an ${\tt MEC}_v \in 
\Theta_t$ containing $q$, we proceed towards the left child of $t$, otherwise we proceed 
towards the right child.
\begin{lemma}\label{lz}
The search with $q$ must stop at a node $t$ of {\tt T}, and outputs a vertex $v$ in the plane graph $G_t$ associated with $t$ such that (i) $q \in {\tt MEC}_v$ and (ii) $C_q$ is centered on some edge of $G_t$.
\end{lemma}
\begin{proof}
Because {\tt MEC}s on Voronoi vertices are circumcircles of triangles in the Delaunay triangulation of $P$,  the union of these {\tt MEC}s covers the entire convex hull region. Since every ${\tt MEC}$ centered on a Voronoi vertex of $G$ is a separator for some node in {\tt T}, the proof of (i) follows. 

Suppose some {\tt MEC} $C'$ centered somewhere outside $G_t$ encloses $q$.
From the planar decomposition of $G$ down to $G_t$, it is clear that the (collective) neighborhood $\Gamma(G_t)$ of $G_t$ consists of vertices that appear in the separator vertices associated with some ancestor of $t$. Therefore, by the Unique Path Lemma, there must exist a vertex $v^* \in \Gamma(G_t)$ such that the {\tt MEC} centered on $v^*$ encloses $q$. Since $v^*$ is associated with some ancestor $t^*$ of $t$, the search path in ${\tt T}$ must have stopped at $t^*$  instead of coming all the way to $t$, which establishes a contradiction.
\qed
\end{proof}

Now that Lemma\ \ref{lz} is established, we can easily see that if the search path in {\tt T} for a query point $q$ stops at node $t \in {\tt T}$, then the two promises required for ${\tt QiC}_t$ data structure are fulfilled. Therefore, assuming {\tt QiC} is correctly designed (established in Section\ \ref{sec:QiC-1}), we get the following lemma. 
\begin{lemma}\label{lem:qmeccorrect}
Algorithm \ref{alg:QMECPrep} preprocesses a set $P$ of points such that given a query point $q$, Algorithm\ \ref{alg:QMECQuery} can be used to correctly finds $C_q$.
\end{lemma}
In the next subsection we describe both the preprocessing and query phases of {\tt QiC} before proving time and space bounds in subsection\ \ref{cc}.

\subsection{{\tt QiC} data structure for points set} \label{sec:QiC-1}

The {\tt QiC} data structure for the points set case (attached to nodes in {\tt T}) largely mimics the simple polygon case. Note that each $t \in {\tt T}$ has a ${\tt QiC}$ data structure attached to it. We reiterate that, while the query point $q$ is promised to lie in a particular {\tt MEC} in the polygon case, in the points set case, $q$ is promised to lie in at least one of the {\tt MEC}s centered on a node in $W_t$. 
The preprocessing and query algorithms are given as self explanatory pseudocode in Algorithm\ \ref{alg:qicprep} and Algorithm\ \ref{alg:qicquery}, respectively. The time and space bounds of the preprocessing phase follow from the following lemma.

\begin{algorithm}[h!]
\caption{Preprocessing for {\tt QiC} attached to node $t \in {\tt T}$.}
\label{alg:qicprep}
\begin{algorithmic}[1]
\REQUIRE{A node $t \in {\tt T}$ and the plane graph $G_t$ attached to $t$ along with separator vertices $W_t$ and parts $A_t$ and $B_t$.}
\STATE Compute ${\tt R} = \{r \mid \exists \text{ an {\tt MEC} of radius $r$ centered on a vertex of $G_t$}\}$.
\FORALL{$v \in W_t$}
\STATE Compute ${\tt S}_v = \{C \mid \text{$C$ is an {\tt MEC} that fulfills the following}\}$:
\begin{enumerate}
\item[] 1. $C$ overlaps with ${\tt MEC}_v$,
\item[] 2. Radius of $C$ is in ${\tt R}$, and
\item[] 3. Every {\tt MEC} in the unique path from ${\tt MEC}_v$ to $C$ has radius no more than that of $C$.
\end{enumerate}
\STATE Sort $S_v$
\STATE Compute BFS tree $\Lambda_v$ rooted at $v$ with centers of {\tt MEC}s  in ${\tt S}_v$ as nodes.
\STATE For each node $\nu$ in $\Lambda_v$, there is at most one edge $\eta$ incident on $\nu$ such that the {\tt MEC}s centered on $\eta$ strictly grow in size (starting from $\nu$). Mark $\eta$ red.
\ENDFOR
\STATE The {\tt QiC} data structure associated with $t$ consists of  ${\tt S}_v$ and $\Lambda_v$ for all $v \in W_t$.
\end{algorithmic}
\end{algorithm}

\begin{algorithm}[h!]
\caption{Query phase for {\tt QiC} attached to node $t \in {\tt T}$.}
\label{alg:qicquery}
\begin{algorithmic}[1]
\REQUIRE{{\tt QiC} data structure and a query point $q$ that meets the two promises.}
\ENSURE{The largest {\tt MEC} containing $q$ is computed and returned.}
\STATE Using the {\tt PLiCA} data structure $\Phi$ attached to $t$, we find a $v \in W_t$  such that ${\tt MEC}_v$ encloses $q$.
\STATE Perform a binary search on ${\tt S}_v$ to find the radius $r_\text{max}$ of the largest {\tt MEC} in ${\tt S}_v$ that encloses $q$. Also compute  \[\mathcal{C} = \{C \mid (C \in {\tt S}_v) \wedge (\text{radius of } C = r_\text{max}) \wedge (C \text{ encloses $q$ }) \}.\]
\STATE $C_{\text{max}} \leftarrow C$, where $C \in \mathcal{C}$ is chosen arbitrarily.
\FORALL{$C \in \mathcal{C}$}
\STATE Let $c$ be the center of $C$. Recall that $c$ is a node in $\Lambda_v$.
\STATE Let $C^*$ be the largest {\tt MEC} centered on the red edge in $\Lambda_v$ incident on $c$. If no red edge is incident on $c$, assign $C^* \leftarrow C$.
\IF{$C^*$ is larger than $C_{\text{max}}$}
\STATE $C_{\text{max}} \leftarrow C^*$.
\ENDIF
\ENDFOR
\RETURN $C_{\text{max}}$.
\end{algorithmic}
\end{algorithm}

\begin{lemma} \label{lem:BoundingS2}
For any  node $t \in {\tt T}$, any vertex $v$ in $G_t$ and any $r \in {\tt R}$, we define  $S_v^r \triangleq \{C \mid C \in S_v \wedge \text{ and radius of $C$ is $r$}\} $. We claim that $|S_v^r|$ is
bounded by a constant. 
\end{lemma}
\begin{proof}
The key ideas required to prove this lemma have already been discussed in the 
context of 
Lemma~\ref{lem:BoundingS}. Therefore, we limit ourselves 
to making a few  important observations that establish a correspondence between the 
current context (where $P$ is a set of points) to the context of 
Lemma~\ref{lem:BoundingS} (where $P$ is a simple polygon). 

Firstly, observe that all circles in $S_v^r$ must lie in a circle $\chi$ of radius $\rho_v + 2r
$ 
centered at $v$; here $\rho_v$ is the radius of  ${\tt MEC}_{v}$ 
and $r$ is the radius of circles in $S_v^r$ (see Figure \ref{fig:BoundingS_r}).

To make the second observation, consider a circle $C \in S_v^r$ centered at a point 
$c$ 
strictly in the interior of  an edge $e = (v_1, v_2)$.
Without loss of generality,  assume ${\tt MEC}_{v_1}$ is smaller than ${\tt MEC}
_{v_2}$. 
Therefore, $C$ will be no smaller  than ${\tt MEC}_{v_1}$ and 
no larger  than ${\tt MEC}_{v_2}$. Our second observation is that the {\tt MEC}s  
centered on 
$e$ are growing in size in the vicinity of $c$ as we move in 
the direction from $v_1$ to $v_2$. 

To make the third observation, note first that the unique path (as defined in the 
Unique 
Path Lemma) from $v$ to $c$ passes through $v_1$ and not 
through $v_2$. Let $p_1$ and 
$p_2$ be the two points in $P$ that touch $C$. The third observation is that  the 
chord 
$p_1p_2$ intersects the unique path (as defined in the Unique Path Lemma) between 
$v$ 
and $c$ (see Figure
\ref{fig:BoundingS_r} for a similar situation in the context where $P$ is a simple 
polygon).  

The rest of the proof follows from the proof of Lemma~\ref{lem:BoundingS}. \qed
\end{proof}

\begin{lemma}\label{lem:qic}
Algorithms\ \ref{alg:qicprep} and \ref{alg:qicquery} correctly implement the {\tt QiC} data structure.
The preprocessing time and space of the {\tt QiC} attached to node $t$ in {\tt T} is bounded by $O(n_t^{3/2} \log n_t)$ and $O(n_t^{3/2})$, respectively, where $n_t$ is the number of vertices in the plane graph $G_t$ attached to $t$. Queries can be answered correctly in $O(\log n_t)$ time.
\end{lemma}
\begin{proof}
Let $G_t$ be the subgraph attached to a node $t \in 
\tt T$, and $W_t$ be the separator vertices of $G_t$. Recall that $|W_t| = O(
\sqrt{|n_t|})$, where $n_t = |G_t|$. In the {\tt QiC} data structure 
for the node $t$, $|S_v|$ can be $O(n_t)$ for each node in $v \in W_t$, 
and it can be computed in $O(n_t\log n_t)$ time. Thus, the time required 
to create the {\tt QiC} data structure for all the nodes in $W_t$ is 
$O(n_t^{3/2}\log n_t)$. The space requirement is $O(n_t^{3/2})$. 

The correctness argument is similar to the polygon case. In the polygon case, a single path between any two points on the medial axis followed immediately from the fact that the medial axis was a tree. In the current context, the Unique Path Lemma provides a similar unique path. 
\qed
\end{proof}

\subsection{Complexity}\label{cc}

Lemma \ref{lem:qmeccorrect} justifies that our proposed algorithm correctly computes 
the 
largest {\tt MEC} containing the query point $q$ among the points in $P$. The 
following 
lemma establishes the  complexity.

\begin{lemma} \label{xxz}
The preprocessing time and space complexities of the {\tt QMEC} problem are  
$O(n^{\frac{3}{2}} \log^2 n)$ and  $O(n^{\frac{3}{2}} \log n)$,
respectively. The query can be answered in $O(\log^2n)$ time.
\end{lemma}
\begin{proof}
The preprocessing consists of the following steps:
\begin{itemize}
\item[$\bullet$] Constructing the tree $\tt T$. At each node $t$ of 
$\tt T$, we need to compute the separator vertices among the set of 
vertices of the Voronoi subgraph $G_t$ corresponding to the node $t$. 
The time complexity for this computation is $O(|V_t|)$, where $|V_t|$ 
denotes the number of vertices in $G_t$. Since the total number of 
vertices at each level of $\tt T$ is $O(n)$, the total time spent for 
computing the separator vertices at all nodes in each level of $\tt T$ 
is $O(n)$.  Since the height of {\tt T} is at most $O(\log n)$,  
the total time for constructing it is $O(n\log n)$.
\item[$\bullet$]  For each node $t$ in {\tt T}, we need to construct the {\tt QiC} data structure. The subgraphs associated with each node in any particular level of {\tt T} are disjoint. Therefore, as a consequence of Lemma\ \ref{lem:qic}, the preprocessing time and space required for {\tt QiC} data structures associated with nodes in any particular level is $O(n^{3/2} \log n)$ and $O(n^{3/2})$, respectively. Since $\tt T$ can have at most 
$O(\log n)$ levels, the preprocessing time and space complexities follow .
\item[$\bullet$] For each node $t$, we also will need to compute two  {\tt PLiCA} data structures $\Phi$ and $\Theta$, but the time and space complexities of the {\tt QiC} data structures dominate the complexities of computing the {\tt PLiCA} data structures.
\end{itemize}

While querying with a point $q$,  searching in the {\tt PLiCA} data 
structures $\Phi$ for each node  in the search path of ${\tt T}$  will take $O(\log n)$ time. We  have 
to traverse a path of length at most $O(\log n)$ to get to a node $t$ in {\tt T} such that  there is a vertex $v$ of 
$G_t$ such that $q\in {\tt MEC}_v$. Thus, traversing $\tt T$ needs 
$O(\log^2n)$ time. Finally searching in $S_v$ and $\Lambda_v$ to get $C_q$ needs another 
$O(\log n)$ time (see Theorem \ref{th:th}). 
\qed
\end{proof}

\subsection{Improving the query time}
We now show that a minor tailoring of the data structure reduces 
the query time to\\
 $O(\log n \log\log n)$, while maintaining the same preprocessing 
time and space.

\subsubsection{Data structure.}
After computing the planar separator tree $\tt T$, each {\tt MEC} $C$ centered on a vertex in $G$ is 
attached with 
\begin{itemize}
\item an $\tt id$, which is the {\it level} of $\tt T$ in which $C$ 
belongs as a separator {\tt MEC}, and
\item a pointer to the node $t$ in {\tt T} such that $C$ belongs to the separator 
vertices of 
$G_t$.
\end{itemize}
 Next, we create an array $\Gamma$ of $O(\log n)$ data structures as follows. Each 
$
\Gamma_i$  is a {\tt PLiCA} data structure constructed 
with the set of {\tt MEC}s with $\tt id$ ranging from 
$1$ to $i$, i.e., root to the level $i$.   

\subsubsection{Query.}
While querying with a point $q$, we conduct a binary search on the array $\Gamma$ 
of 
data structures to find $\Gamma_i$ such that there is an {\tt 
MEC} $C^*$ in $\Gamma_i$ that contains $q$, but no {\tt MEC} in $\Gamma_{i-1}$ 
that 
contains $q$. Let $t^*$ the node in {\tt T} where  the center of $C^*$ is a separator vertex. We now perform a {\tt QiC} query on $t^*$ and 
report the result of that {\tt QiC} query as the required {\tt MEC} $C_q$.
\begin{theorem}
The improvement described in this subsection is correct and its preprocessing time 
and 
space complexities for the {\tt QMEC} problem are  
$O(n^{3/2} \log^2 n)$ and  $O(n^{3/2} \log n)$ 
respectively. Each query can be answered in $O(\log n\log\log n)$ time.
\end{theorem}
\begin{proof}
The correctness follows from the fact that there is no {\tt MEC} with id smaller than 
$C^*$ 
that contains $q$, but $C^*$ in fact contains $q$. Therefore, 
a {\tt QiC} on $C^*$ indeed gives us the required {\tt MEC} $C_q$.

Each $\Gamma_i$ requires $O(n \log n)$ time and $O(n)$ space. Therefore, to 
construct 
$\Gamma$, we require $O(n \log^2 n)$ time and $O(n \log n)
$ space, which are subsumed in the bounds established in Lemma~\ref{xxz} to 
construct $
\tt T$.

In the query phase, each {\tt PLiCA} query on any element of $\Gamma$ requires 
$O(\log 
n)$ time and the binary search over all elements of $
\Gamma$ requires $O(\log \log n)$ such {\tt PLiCA} queries, thereby requiring $O(\log 
n 
\log \log n)$ time overall. The {\tt QiC} query requires an 
additional $O(\log n)$ time, which is subsumed. 
\end{proof}

\subsection{Achieving $O(\log n)$ Query Time}\label{sec:fastquery}
Here, we shall use Frederickson's $r$-partitioning of planar graphs, stated below, 
to improve the query time complexity to $O(\log n)$. Furthermore, this algorithm is 
simpler in that it does not require us to construct a divide and conquer tree.

\begin{lemma} \label{fedrick}\cite{fed} Given a planar graph $G$ with $n$ vertices 
with a planar embedding and a parameter $r$ ($1 \leq r \leq n$),
\begin{itemize}
\item[(a)]  $G$ can be partitioned into $\Theta(\frac{n}{r})$  parts with at most $O(r)$ 
vertices in each part, and a total of $O(\frac{n}{\sqrt{r}})$ boundary vertices over 
all the 
partitions. 
\item[(b)] This partitioning can be computed in $O(n\log n)$ time. 
\end{itemize}
\end{lemma}

We compute the $r$-partitioning of the graph $G$ with $r$ set to $n^{2/3}$. Now, we construct two data structures, $\Upsilon$ and $\Psi$, as stated 
below.
\begin{itemize}
\item[$\Upsilon$:] We 
construct a {\tt PLiCA} data structure and a {\tt QiC} data structure  over all {\tt MEC}s centered on boundary vertices in the $r$-partitioning. 
\item[$\Psi$:] It consists of a {\tt PLiCA} data structures with the set of {\tt MEC}s 
that correspond to the internal vertices of all the partitions.
Furthermore, for each partition $j$, we construct a {\tt QiC} data structure limited to partition $j$ on the {\tt MEC}s centered on the internal vertices of partition $j$.
\end{itemize}
For a given query point $q$, we first search in $\Upsilon$ to check whether there exists 
a boundary {\tt MEC} that contains $q$. Here two cases may arise:
\begin{description}
\item[Case 1:]  If an {\tt MEC} $C \in \Upsilon$ encloses $q$, then we search in the {\tt QiC} 
data structure attached with $\Upsilon$ to identify the largest {\tt MEC} containing $q$.
\item[Case 2:] Otherwise, we search in $\Psi$ to identify an {\tt MEC} $C'$ that 
contains $q$. We also find the partition $j$ on which $C'$ is centered. We  then search in the {\tt QiC} data structure attached with $j'$ to 
identify the largest {\tt MEC} containing $q$.
\end{description}
\begin{lemma} 
The above algorithm correctly identified the largest {\tt MEC} containing $q$. The 
preprocessing time, space and query time complexities of this algorithm are 
$O(n^{5/3} 
\log n)$, $O(n^{5/3})$ and $O(\log n)$, respectively. 
\end{lemma}
\begin{proof}
The correctness of the algorithms follows from the following argument. During the 
query, 
if Case (i) arises the algorithm produces the correct result since the {\tt QiC} data 
structure attached to $\Upsilon$ is built on the {\tt MEC}s corresponding to all 
the vertices in $G$.  If Case (ii) arises, and $q$ lies in an {\tt MEC} $C'$ of the $j$-th 
partition, then it implies that $q$ lies in some {\tt MEC} in the proper interior of the 
$j$-th partition. Thus, the largest {\tt MEC} containing $q$ is surely an {\tt MEC} $C^*$ 
of the $j$-th partition. Thus the {\tt QiC} of $C'$ constructed with the {\tt 
MEC}s in the $j$-th partition only is sufficient to obtain $C_q$.      

Now, we justify the complexity results of the algorithm. The total size of the {\tt QiC} 
data structures in $\Upsilon$ is $O(\frac{n^2}{\sqrt{r}})$, and these are constructed in 
$O(\frac{n^2}{\sqrt{r}}\log n)$ time. The total size of the {\tt QiC} data structures for 
all the {\tt MEC}s in the $j$-th partition of $\Psi$ is $O(r^2)$, and these are 
constructed 
in $O(r^2\log r)$ time. Since, we have at most $O(\frac{n}{r})$ partitions, the total 
space 
and time required to construct $\Psi$ is $O(nr)$ and $(nr\log r)$, respectively. Thus, 
the 
total preprocessing space and time complexities are $O(\frac{n^2}{\sqrt{r}}+nr)$ and 
$O((\frac{n^2}{\sqrt{r}}+nr)\log n)$, respectively. Choosing $r=O(n^{2/3})$, the 
preprocessing 
time and space complexity results follow.

The query time complexity follows from the fact that the search in the {\tt PLiCA} of 
both 
$\Upsilon$ and $\Psi$ take $O(\log n)$ time, and the search in the {\tt QiC} data structure 
of 
exactly one {\tt MEC} needs another $O(\log n)$ time in the worst case. 
\qed
\end{proof}
 
\noindent {\bf Acknowledgments.} We are grateful to Samir Datta and
Vijay Natarajan for their helpful suggestions and ideas.
We are also thankful to Subir Ghosh for providing the environment to
carry out this work.  
Finally, we are grateful to the anonymous referees for providing insightful comments and suggestions.


\begin{thebibliography}{88}
\bibitem{ADMNRS10} J. Augustine, S. Das, A. Maheshwari, S. C. Nandy,
S. Roy, and S. Sarvattomananda, {\it Recognizing the largest empty circle and
axis-parallel rectangle in a desired location.}, {\sf Technical Report:
http://arxiv.org/abs/1004.0558}, 2010.
\bibitem{AGSS} A. Aggarwal, L. J. Guibas, J. Saxe and P. W. Shor, {\it
A linear time algorithm for computing the Voronoi diagram of a convex
polygon}, Proc. of the {\sf 19th Annual ACM Symposium on Theory of
Computing}, pp. 39-45, 1987.
\bibitem{APS10} J. Augustine, B. Putnam, and S. Roy, {\it Largest
empty circle centered on a query line}, {\sf Journal of Discrete
Algorithms}, vol. 8, pp. 143-153, 2010.
\bibitem{AS} A. Aggarwal, S. Suri, {\it Fast algorithms for computing
the largest empty rectangle}, Proc. of the {\sf 3rd Annual Symposium
on Computational Geometry}, pp. 278 - 290, 1987. 
\bibitem{AS1} B. Aronov and M. Sharir, {\it Cutting circles into
pseudo-segments and improved bounds for incidences}, {\sf Discrete
Computational Geometry}, vol. 28, pp. 475-490, 2000.
\bibitem{BF04} M. A. Bender, and M. Farach-Colton, {\it The level
ancestor problem simplified}, {\sf Theoretical Computer Science},
vol. 321, pp. 5-12, 2004. 
\bibitem{BCDUY00} J. Boissonnat, J. Czyzowicz, O. Devillers,
J. Urrutia, and M. Yvinec, {\it Computing largest circles
separating two sets of segments}, {\sf International Journal of  
Computational Geometry and Applications}, vol. 10, pp. 41-54, 2000.
\bibitem{BCOY01} J. Boissonnat, J. Czyzowicz, O. Devillers, and
M. Yvinec, {\it Circular separability of polygons}, {\sf Algorithmica},
 vol. 30, pp. 67-82, 2001.
\bibitem{BU} R. P. Boland and J. Urrutia, {\it Finding the largest
axis aligned rectangle in a polygon in $O(n \log n)$ time}, Proc.
of the {\sf Canad. Conf. on Computational Geometry}, pp. 41-44, 2001.
\bibitem{CND} J. Chaudhuri, S. C. Nandy, S. Das, {\it Largest empty
rectangle among a point set}, {\sf  Journal of  Algorithms}, vol. 46,
pp. 54-78, 2003.
\bibitem{CSW99} F. Y. L. Chin, J. Snoeyink, and C. A. Wang, {\it
Finding the medial axis of a simple polygon in linear time}, {\sf
Discrete Computational Geometry}, vol. 21, pp. 405-420, 1999.
\bibitem{EGLM2003} J. Edmonds, J. Gryz, D. Liang, and R. J. 
Miller, {\it Mining for empty spaces in large data sets}, 
{\sf Theoretical Computer Science}, vol. 296(3), pp. 435-452, 2003. 

\bibitem{Fisk86} S. Fisk, { \it Separating point sets by circles, and the
recognition of digital disks}, {\sf Transactions on Pattern Analysis and Machine
Intelligence}, vol. 8, pp. 554-556, 1986.
\bibitem{fed} G. N. Frederickson, {\it Fast algorithms for shortest 
paths in planar graphs, with applications}, {\sf SIAM J. on Computing}, 
vol. 16, pp. 1004-1022, 1987.
\bibitem{HT91} W. L. Hsu and K. H. Tsai, {\it Linear time algorithm
on circular-arc-graphs}, {\sf Information Processing Letters}, vol.
40, pp. 123-129, 1991.


\bibitem{IIM} H. Imai, M. Iri and K. Murota, {\it Voronoi diagram
in the Laguerre geometry and its applications}, {\sf SIAM Journal on
Computing}, vol. 14: 93-105, 1985.
\bibitem{JP} R. Janardan and F. P. Preparata, {\it Widest  
corridor problem}, {\sf Nordic Journal on Computing}, vol. 1, 
pp. 231-245, 1994.
\bibitem{jordon} C. Jordan, {\it Sur les assemblages de lignes}, {\sf 
Journal fur die Reine und Angewandte Mathematik} vol. 70, pp. 185-190,
1869.
\bibitem{KMNS12} H. Kaplan, S. Mozes, Y. Nussbaum, and M. Sharir, {\it Submatrix maximum queries in Monge matrices and Monge partial matrices, and their applications}, in {\sf Proceedings of the Symposium on Discrete Algorithms}, 2012,  pp. 338-355.
\bibitem{KS12} H. Kaplan and M. Sharir, {\it Finding the Maximal Empty Disk Containing a Query Point}, in {\sf Proceedings of the Symposium on Computational Geometry}, 2012.
\bibitem{Krik83} D. G. Kirkpatrick, {\it Optimal search in planar
subdivisions}, {\sf SIAM Journal on Computing}, vol. 12, pp. 28-35,
1983.
\bibitem{LT79} R. Lipton and R.E. Tarjan, {\it A separator theorem for 
planar graphs}, {\sf SIAM Journal on Applied Mathematics},
vol. 36, pp. 177-189, 1979.
\bibitem{LKH1997} B. Liu, L. Ku, and W. Hsu, {\it Discovering interesting holes in data}, in {\sf Proceedings of the Fifteenth international joint conference on Artifical intelligence}, pp. 930-935, 1997.

\bibitem{MN} K. Mehlhorn and S. Naher, {\it  Dynamic fractional
cascading}, {\sf Algorithmica}, vol. 5, pp. 215-241, 1990.
\bibitem{NHL} A. Naamad, W.-L. Hsu and D. T. Lee,  {\it On the 
maximum empty rectangle problem}, {\sf Discrete Applied 
Mathematics}, vol. 8, pp. 267-277, 1984.
\bibitem{NSB} S. C. Nandy, A. Sinha, B. B. Bhattacharya, {\it 
Location of the largest empty rectangle among arbitrary obstacles}, 
Proc. of the {\sf 14th Annual Conf. on Foundations of Software
Technology 
and Theoretical Computer Science}, LNCS-880, pp. 159-170, 1994.
\bibitem{OKM86} J. O'Rourke, S. Kosaraju and N. Megiddo, {\it 
Computing circular separability}, {\sf Discrete \& Computational 
Geometry}, vol. 1, pp. 105-113, 1986.
\bibitem{P1977} F. P. Preparata, {\it The Medial Axis of a Simple Polygon}, In {\sf Mathematical Foundations of Computer Science}, pp 443-450, 1977.
\bibitem{PS} F. P. Preparata and M. I. Shamos, {\sf Computational
Geometry: An Introduction}, Springer, 1975.
\bibitem{SarTar86} N. Sarnak and R. E. Tarjan, {\it Planar point
location using persistent search trees}, {\sf Communications of the
ACM} vol. 29, pp. 669-679, 7 July 1986.
\bibitem{TT} H. Tamaki and T. Tokuyama, {\it How to cut
pseudo-parabolas into segments}, {\sf Discrete
Computational Geometry}, vol. 19, pp. 265-290, 1998.
\bibitem{T} G. Toussaint, {\it  Computing largest empty circles with
location constraints}, {\sf International Journal of Parallel
Programming}, vol. 12, pp. 347-358, 1983.
\end{thebibliography}
\end{document}